\documentclass[journal]{IEEEtran}

\usepackage{fancyhdr}
\usepackage{latexsym}
\usepackage{graphicx,color}
\usepackage{shadow,epsf,amsthm,amssymb,amsmath}
\usepackage{enumerate}
\usepackage{setspace}                           
\usepackage{epstopdf}
\usepackage{caption}

\usepackage{tikz}	
\usetikzlibrary{backgrounds,fit,decorations.pathreplacing}  
\usetikzlibrary{automata} 

\usepackage{lipsum}

\usepackage[colorinlistoftodos,prependcaption]{todonotes}

\newtheorem{definitionenv}{Definition}
\newtheorem{lemmaenv}[definitionenv]{Lemma}
\newtheorem{theoremenv}[definitionenv]{Theorem}
\newtheorem{corollaryenv}[definitionenv]{Corollary}
\newtheorem{propositionenv}[definitionenv]{Proposition}
\newtheorem{conjectureenv}[definitionenv]{Conjecture}
\newtheorem{exampleenv}{Example}
\newtheorem{app-lemmaenv}[section]{Lemma}

\newenvironment{definition}{\begin{definitionenv}\rm}{\end{definitionenv}}
\newenvironment{lemma}{\begin{lemmaenv}\rm}{\end{lemmaenv}}
\newenvironment{theorem}{\begin{theoremenv}\rm}{\end{theoremenv}}
\newenvironment{corollary}{\begin{corollaryenv}\rm}{\end{corollaryenv}}
\newenvironment{example}{\begin{exampleenv}\rm}{\end{exampleenv}}
\newenvironment{proposition}{\begin{propositionenv}\rm}{\end{propositionenv}}
\newenvironment{conjecture}{\begin{conjectureenv}\rm}{\end{conjectureenv}}
\newenvironment{app-lemma}{\begin{app-lemmaenv}\rm}{\end{app-lemmaenv}}

\newcommand{\bd}{\begin{definition}}
\newcommand{\ed}{\end{definition}}
\newcommand{\bl}{\begin{lemma}}
\newcommand{\el}{\end{lemma}}
\newcommand{\elp}{\hspace*{\fill} $\Box$
                 \end{lemma}}
\newcommand{\bt}{\begin{theorem}}
\newcommand{\et}{\end{theorem}}
\newcommand{\etp}{\hspace*{\fill} $\Box$
                 \end{theorem}}
\newcommand{\bc}{\begin{corollary}}
\newcommand{\ec}{\end{corollary}}
\newcommand{\ecp}{\hspace*{\fill} $\Box$
                 \end{corollary}}
\newcommand{\bcj}{\begin{conjecture}}
\newcommand{\ecj}{\end{conjecture}}

\newcommand{\be}{\begin{example}}
\newcommand{\ee}{\end{example}}
\newcommand{\eep}{\hspace*{\fill} $\Box$
                 \end{example}}
\newcommand{\bp}{\begin{proposition}}
\newcommand{\ep}{\end{proposition}}
\newcommand{\epp}{
                 \end{proposition}}

\newcommand{\bra}[1]{\langle#1|}
\newcommand{\ket}[1]{|#1\rangle}

\newcommand{\Tr}[1]{\text{Tr}\left(#1\right)}
\newcommand{\wt}[1]{\text{wt}\left(#1\right)}

\newcommand{\cC}{{\cal C}}

\newcommand{\cF}{{\cal F}}
\newcommand{\cG}{{\cal G}}
\newcommand{\cH}{{\cal H}}

\newcommand{\cL}{{\cal L}}
\newcommand{\cM}{{\cal M}}

\newcommand{\cQ}{{\cal Q}}
\newcommand{\cS}{{\cal S}}

\newcommand{\cV}{{\cal V}}

\newcommand{\mZ}{{\mathbb Z}}

\newcommand{\odd}{\scriptsize  \mbox{odd}}
\newcommand{\even}{\scriptsize \mbox{even}}

\DeclareMathOperator*{\argmin}{\arg\min}
\begin{document}

\date{}
\title{\Large\bf  Linear Programming Bounds for Entanglement-Assisted Quantum Error-Correcting Codes by Split Weight Enumerators}

\author{Ching-Yi Lai, \emph{Member, IEEE}, and Alexei Ashikhmin, \emph{Fellow, IEEE}
\thanks{
This work
was presented in part at the IEEE International Symposium on Information Theory
2017.

CYL is with the Institute of Information Science, Academia Sinica,
No 128, Academia Road, Section 2
Nankang, Taipei 11529, Taiwan.
(email: cylai0616@iis.sinica.edu.tw)

AA is with the Bell Laboratories,
Nokia, 600 Mountain Ave,
Murray Hill, NJ 07974.
(email: alexei.ashikhmin@nokia-bell-labs.com)
}
}

\maketitle

\thispagestyle{empty}

\begin{abstract}

Linear programming approaches have been applied to derive upper bounds on the size of classical and quantum codes.
In this paper, we derive similar results for general quantum codes with entanglement assistance by considering a type of split weight enumerator.
After deriving the MacWilliams identities for these   enumerators, we are able to prove algebraic linear programming bounds, such as the Singleton bound, the Hamming bound, and the first linear programming bound. Our Singleton bound and Hamming bound are more general than the previous bounds for entanglement-assisted quantum stabilizer codes.
In addition, we show that the first linear programming bound improves the Hamming bound when the relative distance is sufficiently large.

 On the other hand,
  we obtain additional constraints on the size of Pauli subgroups for quantum codes, which allow us to improve the linear programming bounds on the minimum distance of quantum codes of small length. In particular, we show that there is no $[[27,15,5]]$ or $[[28,14,6]]$   stabilizer code.
We also discuss the existence of some entanglement-assisted quantum stabilizer codes with maximal entanglement.
As a result, the  upper and lower bounds on the minimum distance of maximal-entanglement quantum stabilizer codes  with length up to $20$
are significantly improved.
\end{abstract}

\section{Introduction}

The theory of quantum error-correcting codes has been developed over two decades
for the purpose of protecting  quantum information from noise in computation or communication~\cite{Shor95,EM96,KL97,Ste96,Got97,CRSS97,CRSS98,NC00,ABKL00a,ABKL00b}.
Classical coding techniques, such as code constructions, encoding and decoding procedures, have been generalized to the quantum case in the literature~\cite{Gaitan08,LB13}.
An important question in coding theory is to determine  how much redundancy is needed so that a certain amount of errors can be tolerated~\cite{MS77}.
The error-correcting capability of a code is usually quantified by the notion of its \emph{minimum distance}, which can be determined by the \emph{weight enumerator} of the code.
Delsarte applied the duality theorem of linear programming to find universal upper bounds on the size of classical codes~\cite{Del73}.
A key step is to use the MacWilliams identities that relate the weight enumerators of a code and its dual code~\cite{MS77}.
These linear programming bounds have been generalized to  the quantum case by Ashikhmin and Litsyn~\cite{AL99},
based on the existence of quantum MacWilliams identities~\cite{SL96,Rains96b}.
We would like to extend these results to the case of quantum codes with entanglement assistance.

Entanglement-assisted (EA) quantum stabilizer codes are a coding scheme with an additional resource--entanglement--shared between the sender and receiver~\cite{BDM06}.
These codes have some advantages over standard stabilizer codes~\cite{HBD08,HYH11,FCVDT10,WHZ14}.
In the scheme of an EA stabilizer code, it is assumed that  the receiver's qubits are error-free, which makes the analysis of error correction slightly more complicated.
Previously, Lai, Brun, and Wilde discussed the MacWilliams identities for the case of EA stabilizer codes \cite{LBW13,LBW10},
which naturally follow from the MacWilliams identities for orthogonal groups~\cite{MS77}.
Two dualities (see Eqs. (\ref{eq:EA stabilizerC_duality1}) and (\ref{eq:EA stabilizerC_duality2}) below) are used to find linear programming bounds on the minimum distance of small EA stabilizer codes.
However, only one duality (see Eq. (\ref{eq:quantum duality}) below) appears in   standard stabilizer codes, and thus we cannot directly apply the method in \cite{AL99} to obtain algebraic linear programming bounds for EA stabilizer codes.
In addition, examples of  \emph{nonadditive} EA quantum codes have recently been found in \cite{SHB11}, which apparently do not fit the EA stabilizer formalism~\cite{BDM06}.
 Consequently, previous upper bounds for EA stabilizer codes, such as the Singleton bound and Hamming bound, do not work for nonadditive EA quantum codes.

In this paper, we will first define general EA quantum codes and discuss their error correction conditions.
Two split weight enumerators  of an EA quantum code are defined accordingly and they are proved to obey a MacWilliams identity, similarly to nonadditive quantum codes \cite{SL96,Rains96b}.
These weight enumerators bear sufficient information about the error correction conditions of the EA code (see Theorem~\ref{thm:EA MI}).
Recently the notion of \emph{data-and-syndrome correction codes} is introduced~\cite{ALB16}, and algebraic linear programming bounds for these codes are derived from a type of weight enumerators over the product space of $\mathbb{F}_4$ and $\mathbb{Z}_2$.
We use similar techniques to derive algebraic linear programming bounds on the size of EA quantum codes of given length, minimum distance, and entanglement,
and obtain the Singleton bound, the Hamming bound, and the first linear programming (LP-1) bound for \emph{unrestricted} (degenerate or nondegenerate, additive or nonadditive) EA quantum codes.
In addition, we show that the LP-1 bound improves the Hamming bound when the ratio of code distance to code length (relative distance) is sufficiently large.
Also the previous Singleton bound for EA stabilizer codes in \cite{BDM06} does not work for large minimum distance (i.e., $d>\frac{n+2}{2}$; see definitions below).
Recently examples of EA stabilizer codes that violate  the Singleton bound have been constructed~\cite{Gra17}. We will provide a refined Singleton bound for the general case.

In the case of EA stabilizer codes, the MacWilliams identities for split weight enumerators provide more constraints in the linear program of an EA stabilizer code than those in \cite{LBW13}, and thus the resulting upper bound on the minimum distance could be potentially tighter for EA stabilizer codes of small length.
Other than that, we also find \emph{additional constraints} on the weight distributions of EA stabilizer codes.
Rains introduced the notion of \emph{shadow} enumerator for additive quantum codes. In the case of a quantum stabilizer code,
a type of MacWilliams identity between the weight enumerators of a stabilizer  group and its shadow (see Eq.~(\ref{eq:MI_shadow})) provides additional constraints on the weight distribution of its stabilizer group and hence
the linear programming bound  can be improved~\cite{Rains96}.
However, this method cannot be applied to non-Abelian groups, such as the normalizer group of a stabilizer code or the \emph{simplified} stabilizer group of an EA stabilizer code.
 We will derive additional constraints on the weight enumerators of non-Abelian Pauli groups and improve the linear programming bounds on the minimum distance of small EA stabilizer codes. In particular, this helps to exclude the existence of $[[27,15,5]]$ and $[[28,14,6]]$   quantum  codes.
As for EA stabilizer codes, the improved linear programming bounds rule out the existence of several EA stabilizer codes.
We also prove the nonexistence of certain codes and construct several EA stabilizer codes with previously unknown parameters.
To sum up, the table of upper and lower bounds on the minimum distance of \emph{maximal entanglement} EA stabilizer codes of length up to $20$ given in~\cite{LBW10} (with lower bounds improved in \cite{LLGF15}) is greatly improved in Table~\ref{tb:Bounds}.

This paper is organized as follows. Preliminaries are given in the next section.
In Sec.~\ref{sec:nonadditive codes} we discuss general EA quantum codes and their properties, including two split weight enumerators and their MacWilliams identities.
In particular, we prove a Gilbert-Varshamov type lower bound for the case of EA stabilizer codes.
Then we derive algebraic linear programming bounds in Sec.~\ref{sec:LP bounds}, including the Singleton-type, Hamming-type, and the first-linear-programming-type  bounds.
We will compare the Hamming bound and the first linear programming bound in Subsec.~\ref{sec:LP1}.
The linear programming bounds for small EA stabilizer codes are discussed in Sec.~\ref{sec:small codes}, including
the additional constraints,
and nonexistence and existence of certain EA stabilizer codes.
Then the discussion section follows.

\section{Preliminaries}

In this section we give notation and briefly introduce Pauli operators, quantum error-correcting codes, the MacWilliams identities for orthogonal groups, and properties of the Krawtchouk polynomials.

\subsection{Pauli Operators}
A single-qubit state space is a two-dimensional complex Hilbert space $\mathbb{C}^2$,
and a multiple-qubit state space is simply the tensor product space of single-qubit spaces.
The Pauli matrices
$$I_2=\begin{bmatrix}1 &0\\0&1\end{bmatrix}, X=\begin{bmatrix}0 &1\\1&0\end{bmatrix},  Z=\begin{bmatrix}1 &0\\0&-1\end{bmatrix}, Y=iXZ$$
form a basis of the linear operators on a single-qubit state space.
Let  $$\overline{\cG}_n=\{M_1\otimes M_2\otimes \cdots \otimes M_n: M_j\in\{I_2,X,Y,Z\} \},$$ which is a basis of the linear operators on the $n$-qubit state space $\mathbb{C}^{2^n}$.
Let $$\cG_n= \{i^eg: g\in \overline{\cG}_n,e\in\{0,1,2,3\}\}$$ be the $n$-fold Pauli group.
The weight of $g=i^e M_1\otimes M_2\otimes \cdots \otimes M_n \in {\cG}_n$ is the number of  $M_j$'s that are nonidentity matrices and \textcolor[rgb]{0.0,0.00,0.00}{is denoted by $\wt{g}$}.
Note that all elements in $\bar{\cG}_n$ have only eigenvalues $\pm 1$ and they either commute or anticommute with each other.

Since the overall phase of a quantum state is not important,
it suffices to consider error operators in $\overline{\cG}_n$.
For a subgroup $\cV\subset \cG_n$, we define
\begin{align}
\overline{\cV}=\{g\in \overline{\cG}_n:  i^e g\in \cV \text{ for } e\in\{0,1,2,3\}\}. \label{eq:quotient_group}
\end{align}
That is,  $\overline{\cV}$ is the collection of the elements in $\cV$ with coefficient $+1$.
Often it is convenient to consider the quantum coding problem in terms of binary strings~\cite{CRSS97}.
For two binary $n$-tuples $u,v\in \mathbb{Z}_2^n$, define
 \[
 Z^{u}X^{v}= \bigotimes_{j=1}^n Z^{[u]_j}X^{[v]_j},
 \]
 where $[u]_j$ denotes the $j$th entry of $u$.
Thus any element $g\in {\mathcal{G}}_n$ can be expressed as $g= i^e Z^{u}X^{v}$ for some $ e \in \{0,1,2,3 \}$ and $u,v\in \mathbb{Z}_2^n$.
 For example, we may denote $X\otimes Y\otimes Z$ by $Z^{011}X^{110}$ up to some phase.
 Let \textcolor[rgb]{0.00,0.00,0.00}{$\mathbb{I}$ denote the identity operator of appropriate dimensions.}
  We will use the notation $X_j$ to denote the operator $I_2\otimes \cdots\otimes I_2\otimes X\otimes I_2\otimes \cdots\otimes I_2$ (of appropriate dimensions)
 with an $X$ on the $j$th position and identities on the others, and similarly for $Z_j$ and $Y_j$.
We  define a homomorphism $\tau:{\cG}_n\rightarrow\mZ_2^{2n}$ by
\begin{align}
    i^e Z^uX^v\mapsto  (u , v). \label{eq:tau}
\end{align}
Define an inner product in ${\cG}_n$ by
\begin{align*}
    \langle g,h\rangle_{{\cG}_n} =&\sum_{i=1}^n ([u_1]_i[v_2]_i+[v_1]_i[u_2]_i),
\end{align*}
where $g=i^{e}Z^{u_1}X^{v_1}, h=i^{e'}Z^{u_2}X^{v_2}\in{\cG}_n$ and the addition is considered in $\mathbb{Z}_2$.
Then $\langle g,h\rangle_{{\cG}_n}=0$ if they commute,
 and  $\langle g,h\rangle_{{\cG}_n}=1$, otherwise.

More details on Pauli operators and their properties can be found in \cite{NC00}.

\subsection{Quantum Error-Correcting Codes}

An $((n,M,d))$ quantum code $\cQ$ of length $n$ and size $M$ is an $M$-dimensional linear subspace of  the $n$-qubit state space $\mathbb{C}^{2^n}$.
According to  the error correction conditions~\cite{BDSW96,KL97},
an error $E\in \overline{\cG}_n$  is \emph{detectable} if and only if $\bra{v}E\ket{w}=0$ for any orthogonal  $\ket{v}, \ket{w}\in \cQ$.
The parameter $d$ is called the \emph{minimum distance} of $\cQ$
such that any error $E\in \overline{\cG}_n$ of $\wt{E}\leq d-1$ is {detectable}.

Suppose $\cS=\langle g_1, \dots, g_{n-k}\rangle$ is an Abelian subgroup of ${\cG}_n$, where $g_j$ are independent generators of $\cS$, such that the minus identity $-\mathbb{I}\notin \cS$.
Then $\cS$ defines a quantum \emph{stabilizer} code 
$$\cC(\cS)=\{\ket{\psi}\in \mathbb{C}^{2^n}: g\ket{\psi}=\ket{\psi}, \forall g\in \cS\}$$
of dimension $2^k$.
The vectors $\ket{\psi}\in \cC(\cS)$ are called the \emph{codewords} of $\cC(\cS)$ and the operators $g\in\cS$  are called the \emph{stabilizers} of $\cC(\cS)$.
Quantum stabilizer codes are  
analogues of classical additive codes. As opposed to stabilizer codes, others are called \emph{nonadditive} quantum codes \cite{RHSS97,Rus00,SSW07,YCLO07,CSSZ09,Ouy14,LB13}.

Suppose an error $E\in \overline{\cG}_n$ occurs on a codeword $\ket{\psi}\in \cC(\cS)$.
If $E$ anticommutes with some $g_j$'s, it can be detected by measuring observables $g_j$'s for $E\ket{\psi}$.
Let
\begin{align}
\cS^{\perp}= \{h\in {\cG}_n: \langle h,g\rangle_{{\cG}_n}=0, \ \forall g\in \cS \}, \label{eq:Sperp}
\end{align}
which is the normalizer group  of $\cS$ in ${\cG}_n$.
It is clear that for $E\in \overline{\cS^{\perp}}$, $E$ cannot be detected. Thus the minimum distance $d$ of $\cC(\cS)$ is defined as the minimum weight of any element in
\begin{align}\overline{\cS^{\perp}}\setminus \overline{\cS},\label{eq:quantum duality}
\end{align}
since  errors in $\cS$ are not harmful.  Then $\cC(\cS)$ is called an $[[n,k,d]]$ quantum stabilizer code.
If there exists $g\in \cS$ with $\wt{g}< d$, $\cC(\cS)$ is called \emph{degenerate}; otherwise, it is \emph{nondegenerate}.

The weight enumerator of a group $\cS$ is
\[
W_{\cS}(x,y)=\sum_{w=0}^n B_w x^{n-w}y^w,
\]
where $B_w$ is the number of elements of weight $w$ in $\cS$.
We may simply say that $\{B_w\}$ is the \emph{weight distribution} of $\cS$.
The weight enumerators of an \emph{additive} group $\cV$ and its orthogonal group  $\cV^{\perp}$ are related by the MacWilliams identities~\cite{MS77,LHL14}:
\begin{equation}
W_{{\cV^{\perp}}}(x,y)=\frac{1}{|\cV|}W_\cV(x+3y,x-y), \label{cor:MI_Hamming}
\end{equation}

Thus we have the MacWilliams identities for stabilizer codes.\footnote{More precisely, when we apply the MacWilliams identities for orthogonal groups, we are considering the stabilizer subgroup defined in the quotient of the $n$-fold Pauli group by its center: $\cG_n/\{\pm \mathbb{I},\pm i\mathbb{I}\}$ as in~\cite{LBW13}. Or one can consider the group $\tau(\cS)$ and its orthogonal group in $\mathbb{Z}_2^{2n}$ instead.}
Note that the MacWilliams identities  for nonadditive quantum codes also exist~\cite{SL96,Rains96b}.
As a consequence, linear programming techniques can be applied to find upper bounds on the minimum distance of small quantum stabilizer codes  \cite{CRSS98,LBW13}
or to derive Delsarte's algebraic upper bounds on the dimension of general quantum codes~\cite{AL99}.

More details on definitions and basic properties of quantum codes and weight enumerators can be found in~\cite{NC00,LB13,AL07}.

\subsection{Krawtchouk Polynomials}
The $i$th   quaternary Krawtchouk polynomial is defined as
\begin{align}
K_i(x;n)= \sum_{j=0}^i (-1)^j 3^{i-j}{x\choose j}{n-x \choose i-j}\label{kraw},
\end{align}
which satisfies
\begin{align*}
(1-y)^x(1+3y)^{n-x}= \sum_{i=0}^n K_{i}(x;n) y^i.
\end{align*}
The Krawtchouk polynomials satisfy the following orthogonality relation
\begin{equation}
\label{orth}
\sum_{i=0}^n K_r(i;n)K_i(s;n)=4^n\delta_{r,s}.
\end{equation}
where $\delta_{{  r},{  s}}$ is the Kronecker delta function.
Thus they form a basis of polynomials with finite degree $n$.
More information on Krawtchouk polynomials can be found in \cite{MS77,Lev95,AL99}. 
%
Here we survey some properties of the quaternary Krawtchouk polynomials.
For convenience, sometimes we may simply write $K_i(x)=K_i(x;n)$  when the underlying $n$ is clear from the context.

Let $\{B_w\}$ and  $\{B_w^{\perp}\}$ be  the weight distributions of a stabilizer group $\cS$ and its orthogonal subset $\overline{\cS^{\perp}}$ in $\overline{\cG}_n$, respectively. Thus   (\ref{cor:MI_Hamming}) and (\ref{kraw}) imply
\[
B_i^{\perp}= \frac{1}{|\cS|}\sum_{j=0}^n B_{j} K_{i}(j;n), \ i=0,\dots,n.
\]

A symmetry relation is satisfied by the polynomials
\begin{equation}
\label{sym}
3^i{n \choose i}K_s(i;n)=3^s{n \choose s}K_i(s;n).
\end{equation}
The following property is needed in the proof of Singleton bound later (Sec.~\ref{sec:singleton})
\begin{align}
\sum_{i=0}^n {n-i\choose n-j}K_i(x;n)= 4^j{n-x\choose j}. \label{eq:KP1}
\end{align}

Every polynomial of degree at most $n$ has a unique expansion in the
basis of Krawtchouk polynomials.
If a polynomial $f(x)$ has the expansion
$$f(x)=\sum_{i=0}^t f_i K_i(x) ,$$
then
$$f_i=4^{-n} \sum_{j=0}^n f(j) K_j(i) .$$


The Christoffel-Darboux formula
is of
importance  \cite{Lev95}:
 \begin{align}
 K_{t+1}(x)K_{t}(a)-&K_{t}(x)K_{t+1}(a)\notag\\
=&\frac{4(a-x)} {t+1}3^t{n \choose t}\sum_{i=0}^t\frac{K_i(x)K_i(a)}{3^i{n \choose i }}. \label{ch-dar}
\end{align}

The Krawtchouk polynomials satisfy a recurrence relation,
\begin{align}
(i+1) K_{i+1}(x) = &(3n-2i-4x) K_i(x)\notag\\
& - 3(n-i+1) K_{i-1}(x). \label{rec}
\end{align}

The following equation is derived in~\cite[Lemma 2]{AL99}
\begin{align}
K_i(x)K_j(x)&=\sum_{l=0}^n K_l(x) \sum_{r=0}^{n-l} \alpha(l,i,j,r), \label{eq:pipj}
\end{align}
where
\begin{align}
\alpha(l,i,j,r)=&{l\choose 2l+2r-i-j}{n-l\choose r}{2l+2r-i-j \choose l+r-j}
\nonumber \\
&\times 2^{i+j-2r-l}3^r. \label{eq:alpha}
\end{align}

\textcolor[rgb]{0,0.00,0.00}{
Denote by $r_t$ the smallest root of $K_t(x)$. It is well known that $r_t>r_{t+1}$~\cite{MS77} and that when $t$ grows linearly with $n$ we have
$$
{r_t\over n} ={3\over 4} -{t\over 2n} -{1\over 2} \sqrt{3{t\over n} (n-{t\over n})}  +o(1).
$$
Define
\begin{equation}\label{eq:f(x)}
f(x)=\frac{1}{a-x}
\left\{ K_{t+1}(x)K_{t}(a)-K_{t}(x)K_{t+1}(a)\right\}^2.
\end{equation}
This polynomial yields the first linear
programming bound for classical codes  over $\mathbb{F}_4$ \cite{Aal77},\cite{Lev95}.
To get the first linear programming bound in the asymptotic case, for large $n$,
we choose
$$\frac{t}{n}=\frac{3}{4}-\frac{1}{2}\delta-\frac{1}{2}\sqrt{3\delta(1-\delta)}+o(1),$$
where $\delta=d/n$,  and $a$ so that
$r_{t+1}<a <r_{t}$ and $\frac{K_t(a)}{K_{t+1}(a)}=-1$.}

The Krawtchouk polynomials  also form a basis for multivariate polynomials. In our case, we will need the basis of bivariate polynomials of degrees at most $n,c$ in $x,y$, respectively:
\[
\{ K_i(x;n)K_j(y;c)  \}.
\]
Then
 a polynomial $f(x,y)$ of degree at most $n,c$ in $x,y$ has a unique \emph{Krawtchouk expansion}
\[
f(x,y)= \sum_{i=0}^{n}\sum_{j=0}^{c} f_{i,j}K_i(x;n)K_j(y;c) ,
\]
where
\begin{align}
f_{i,j} =&\frac{1}{4^{n+c}}\sum_{u=0}^n \sum_{v=0}^{c}f(u,v)  K_u(i;n)K_v(j;c) . \label{eq:expansion}
\end{align}

\section{Entanglement-Assisted Quantum Codes}\label{sec:nonadditive codes}

In the following we  consider Pauli operators of the form
$$ E^{A}\otimes F^{B}\in\overline{\cG}_{n+c}$$
for $E^A\in \overline{\cG}_n$, $F^B\in \overline{\cG}_c$,
where the superscripts $A$ and $B$ denote two parties Alice and Bob, respectively. We may implicitly write $E\otimes F$ for simplicity when it is clear from the context.
We will define two \emph{split weight enumerators}  for EA quantum codes that count weights on Alice and Bob's qubits separately,
and then derive their  MacWilliams identities, which is a key step to prove  algebraic linear programming bounds for general EA quantum codes.

Assume Alice and Bob share $c$ pairs of maximally-entangled states $(\ket{00}+\ket{11})/\sqrt{2}$, called \emph{ebits}.
In addition, Alice has other $n-c$ qubits in the state $\ket{0}$.
Then Alice will encode information in her qubits (a total of $n$ qubits) and send them to  Bob through a noisy channel.  Bob's $c$ qubits are assumed to be error-free during the whole process.
Since the qubits of Alice and Bob are entangled, the encoded quantum state lies in an $(n+c)$-qubit state space. (Details of the encoding procedure can be found in~\cite{SHB11}.)
Thus we define an EA quantum code as follows.
\bd \label{def:EAQC}
An $((n,M,d;c))$ EA quantum code $\cQ$ is an $M$-dimensional subspace of  the $(n+c)$-qubit state space $\mathbb{C}^{2^{n+c}}$ such that\\
1) for $\ket{\psi}\in\cQ$,  $\text{Tr}_{A}\left(\ket{\psi}\bra{\psi}\right)=\frac{1}{2^c}\mathbb{I}^B$;\\
2) for  $E^A\in \overline{\cG}_n$ of $\wt{E^A}\leq d-1$,   $E'= E^A\otimes \mathbb{I}^B$ is detectable.
\ed
\noindent
The first condition in Def.~\ref{def:EAQC} ensures that Alice and Bob share $c$ ebits and that the encoding is locally performed by Alice.
The parameter $d$ is called the minimum distance of $\cQ$, which quantifies the maximum weight of a detectable Pauli error.
Denote by $R=\frac{\log_2 M}{n}$ the \emph{code rate}, by $\delta=\frac{d}{n}$ the \emph{relative distance}, and by $\rho=\frac{c}{n}$ the \emph{entanglement-assistance rate}.
If $\cQ$ is defined by a stabilizer group (of ${\cG}_{n+c}$) of size $2^{n+c-k}$, we have $M=2^k$ and $\cQ$ is called an $[[n,k,d;c]]$ EA stabilizer code~\cite{BDM06}.
 If, furthermore, $c=n-k$, $\cQ$ is called a \emph{maximal-entanglement} EA stabilizer code.

Suppose $\cQ$ has an orthonormal basis $\{\ket{\psi_i}\}$ and $P=\sum_i \ket{\psi_i}\bra{\psi_i}$ is the orthogonal projector onto $\cQ$.  The quantum error correction conditions~\cite{BDSW96,KL97}
say that an error operator $E\in \overline{\cG}_{n+c}$ is detectable by $\cQ$ if and only if
\[
\bra{\psi_i}E\ket{\psi_j}=\lambda_E \delta_{i,j}
\]
or
$PEP=\lambda_E P$
for some constant $\lambda_E$ that depends on $E$. \label{lemma:error correction condition}
An error operator $E$ is called a \emph{degenerate} error if $\lambda_E\neq 0$; otherwise, it is \emph{nondegenerate}.

Similarly to the case of quantum codes~\cite{SL96}, we define two split weight enumerators $\{B_{i,j}\}$ and $\{B_{i,j}^{\perp}\}$ of  $\cQ$  by
\begin{align}
B_{i,j}=&\frac{1}{M^2}\sum_{ E\in \overline{\cG}_n, \wt{E}=i\atop E'\in \overline{\cG}_c, \wt{E'}=j  } (\Tr{(E\otimes E') P})^2, \label{eq:Bij}\\
B^{\perp}_{i,j}=&\frac{1}{M}\sum_{ E\in \overline{\cG}_n, \wt{E}=i\atop E'\in \overline{\cG}_c, \wt{E'}=j  } \Tr{(E\otimes E') P(E\otimes E') P}, \label{eq:Bij_perp}
\end{align}
for $i=0,\dots,n$, and $j=0,\dots,c.$
Since $P=\sum_i \ket{\psi_i}\bra{\psi_i}$, we can rewrite $B_{i,j}$ and $B_{i,j}^{\perp}$ as
\begin{align*}
&B_{i,j}=\frac{1}{M^2}\sum_{ E\in \overline{\cG}_n, \wt{E}=i\atop E'\in \overline{\cG}_c, \wt{E'}=j  } \left|\sum_{l} \bra{\psi_l}E\otimes E'\ket{\psi_l}\right|^2, \\
&B^{\perp}_{i,j}=\frac{1}{M}\sum_{ E\in \overline{\cG}_n, \wt{E}=i\atop E'\in \overline{\cG}_c, \wt{E'}=j  } \sum_{l,m} \left| \bra{\psi_l}E\otimes E'\ket{\psi_m}\right|^2.
\end{align*}
Thus for a stabilizer code $\{B_{i,j}\}$ is the distribution of errors that do not corrupt the code space of $\cQ$, while $\{B_{i,j}^{\perp}\}$ is a distribution of errors that are undetectable by $\cQ$.
\bt \label{thm:EA MI}
Suppose $\cQ$ is an $((n,M,d;c))$  EA quantum code with  projector $P$ and split weight enumerators
$\{B_{i,j}\}$ and $\{B_{i,j}^{\perp}\}$, defined in (\ref{eq:Bij}) and (\ref{eq:Bij_perp}), respectively.
Then
\begin{enumerate}[1)]
  \item $B_{0,0}=B_{0,0}^{\perp}=1$; $B_{i,j}^{\perp}\geq B_{i,j}\geq 0$;
\item $B_{0,j}=0$ for $j=1,\dots,c$;
  \item
  \begin{align}
  \left\{
    \begin{array}{ll}
      B_{i,0}=B_{i,0}^{\perp}, & \hbox{ $i=1,\dots,d-1$;} \\
      B_{d,0}<B_{d,0}^{\perp}; &
    \end{array}
  \right.
\label{eq:thm2-2}
  \end{align}
  \item

\begin{equation}
B_{i,j}^{\perp}={M\over {2^{n+c}}}\sum_{u=0}^n\sum_{v=0}^{c}B_{u,v} K_i(u;n) K_j(v;c); \label{cor:MI_3a}
\end{equation}
\begin{equation}
B_{i,j}={1\over {2^{n+c}M}}\sum_{u=0}^n\sum_{v=0}^{c}B_{u,v}^{\perp} K_i(u;n) K_j(v;c). \label{cor:MI_3}
\end{equation}
\end{enumerate}
\et
\begin{proof}
1) The first part is straightforward and the second part is from the Cauchy-Schwartz inequality.

2) By Definition~\ref{def:EAQC}, we have $\text{Tr}_A(P)=\frac{M}{2^c}\mathbb{I}^B$.
Then for $j>0$,
\begin{align*}
B_{0,j}=&\frac{1}{M^2}\sum_{  E\in \overline{\cG}_c, \wt{E}=j  } (\Tr{(\mathbb{I}^A\otimes E) P})^2\\
=&\frac{1}{M^2}\sum_{  E\in \overline{\cG}_c, \wt{E}=j  } \left(\Tr{E \text{Tr}_A(P)}\right)^2=0,
\end{align*}
since $\Tr{E}=0$ for a nonidentity Pauli operator $E$.

3)
We have $B_{i,0}=B_{i,0}^{\perp}$ if and only if $\bra{\psi_l}E^A\otimes \mathbb{I}^{B}\ket{\psi_m}=0 $ for $E\in\overline{\cG}_n$ with $\wt{E}=i$ and for all $l\neq m$.

4)
The proof is similar to that in \cite{SL96}.
The projector $P$  can be expressed as
\begin{align*}
P=&  \sum_{D\in \overline{\cG}_n,  D'\in \overline{\cG}_c} \frac{\Tr{(D\otimes D') P}}{2^{n+c}} D\otimes D',\\
=& \sum_{u=0}^n\sum_{v=0}^c \sum_{D\in \overline{\cG}_n, \wt{D}=u\atop D'\in \overline{\cG}_c, \wt{D'}=v} \frac{\Tr{(D\otimes D') P}}{2^{n+c}} D\otimes D',
\end{align*}
since $\{D\otimes D': D\in \overline{\cG}_n, D'\in \overline{\cG}_c\}$ is a basis of linear operators on $\mathbb{C}^{2^{n+c}}$.
For convenience, we will simply write $D_u\otimes D_v'$ as the index of summation, instead of $$D\in \overline{\cG}_n, \wt{D}=u, D'\in \overline{\cG}_c, \wt{D'}=v,$$ and similarly for $E_i\otimes E_j'$.
By definition,
\begin{align*}
&B^{\perp}_{i,j}\\
=&\frac{1}{M}\sum_{ E_i\otimes E_j' }\sum_{u,v\atop u',v'}\sum_{D_u\otimes D_v'\atop D_{u'}\otimes D_{v'}'}\frac{\Tr{(D_u\otimes D_v' )P}\Tr{(D_{u'}\otimes D_{v'}' )P}}{2^{2(n+c)}}\\
&\cdot \Tr{(E_i\otimes E_j')(D_u\otimes D_v') (E_i\otimes E_j') (D_{u'}\otimes D_{v'}')}\\
\stackrel{(a)}{=}&\frac{1}{M}\sum_{ E_i\otimes E_j' }\sum_{u,v}\sum_{D_u\otimes D_v'}\frac{(\Tr{(D_u\otimes D_v' )P})^2 }{2^{2(n+c)}}\\
&\cdot \Tr{(E_i\otimes E_j')(D_u\otimes D_v') (E_i\otimes E_j') (D_{u}\otimes D_{v}')}\\
\stackrel{}{=}&\frac{M}{2^{n+c}}\sum_{u,v}\sum_{D_u\otimes D_v'}\frac{(\Tr{(D_u\otimes D_v' )P})^2 }{M^2}\\
&\cdot\sum_{ E_i } \frac{\Tr{E_iD_uE_iD_u}}{2^n}\sum_{E_j' }\frac{\Tr{E_j'D_v'E_j'D_{v}}}{2^c}\\
\stackrel{(b)}{=}&\frac{M}{2^{n+c}}\sum_{u,v} B_{u,v} K_{i}(u;n)K_{j}(v;c),
\end{align*}
where $(a)$ is because the trace of a nonidentity Pauli operator is zero.
To prove $(b)$, note that $\Tr{E_iD_u E_iD_u}/2^n=  1$ if $E_i$  and $D_u$ commute; and
$\Tr{E_iD_u E_iD_u}/2^n=  0$, otherwise. The rest is simply to determine the number of $E_i$
that commute   with $D_u$.

Equation~(\ref{cor:MI_3}) is obtained by applying (\ref{orth}) twice to (\ref{cor:MI_3a}).
\end{proof}
\bd
An $((n,M,d;c))$ EA quantum code is called \emph{degenerate} if its split weight enumerator $B_{i,0}>0$ for some $0<i<d$, and \emph{nondegenerate}, otherwise.
\ed

\textbf{Remark:}
The converse of Theorem~\ref{thm:EA MI} is not necessarily true;  given two distributions $\{B_{i,j}\}$ and $\{B_{i,j}^{\perp}\}$ satisfying the conditions 1), 2), 3), and 4) in Theorem~\ref{thm:EA MI},
it may still be the case that no corresponding EA quantum code exists.
In particular, we will derive additional constraints for the case of  {EA stabilizer codes} in Theorem~\ref{thm:new constraints} later.

Usually the existence of a code is shown by construction. The Gilbert-Varshamov bound provides a nonconstructive proof of the existence of classical codes~\cite{MS77}.
 Next we will prove a Gilbert-Varshamov-type bound on the size of {EA stabilizer codes}\footnote{A similar result has been provided in \cite{LBW10}, but the argument there was incomplete.} for fixed $n,d,c$ in Theorem~\ref{thm:EA GV bound}.
We start by defining three  Pauli subgroups associated with an EA stabilizer code.

Let $\cS'$ be an Abelian subgroup of ${\cG}_{n+c}$  generated by
\begin{align*}
&g_1'=g_{1}^A\otimes Z_1^B, \dots, g_c'=g_{c}^A\otimes Z_c^B,\\
&h_1'=h_{1}^A\otimes X_1^B, \dots, h_c'=h_{c}^A\otimes X_c^B,\\
&g_{c+1}'=g_{c+1}^A\otimes \mathbb{I}^B, \dots, g_{n-k}'=g_{n-k}^A\otimes \mathbb{I}^B,
\end{align*}
where $Z_j^B, X_j^B\in {\cG}_c$,  $0\leq c\leq n-k$, $g_j^A, h_j^A \in {\cG}_n$, and
 $g_j$ and $h_j$  satisfy the commutation relations:
\begin{align}
&\langle g_i,g_j\rangle_{{\cG}_n} =0\mbox{ for $i\neq j$}, \label{eq:commutation_1} \\
&\langle h_i,h_j\rangle_{{\cG}_n} =0\mbox{ for $i\neq j$}, \label{eq:commutation_2} \\
&\langle g_i,h_j\rangle_{{\cG}_n} =0 \mbox{ for $i\neq j$},  \label{eq:commutation_3}\\
&\langle g_i,h_i\rangle_{{\cG}_n} =1 \mbox{ for all $i$}.   \label{eq:commutation_4}
\end{align}
(We say that $g_i$ and $h_i$ are  \emph{symplectic partners}.)
Then $\cS'$ defines an $[[n,k,d;c]]$ EA stabilizer code $\cC(\cS')$ by
$$\left\{\ket{\psi}\in \mathbb{C}^{2^{n+c}}: g\ket{\psi}=\ket{\psi}, \forall g\in \cS'\right\}.$$
We can introduce another $2k$ independent generators  $g_{n-k+1}^A\otimes \mathbb{I}^B, \dots, g_n^A\otimes \mathbb{I}^B,
h_{n-k+1}^A\otimes \mathbb{I}^B, \dots, h_n^A\otimes \mathbb{I}^B\in {\cG}_{n+c}$, where $g_j$ and $h_j$ also satisfy the commutation relations.
Since these $2k$ generators commute with the stabilizers in $\cS'$, they operate on the logical level of the encoded states.
Thus we define a \emph{logical} subgroup
$\cL=\langle g_{n-k+1}, \dots, g_n, h_{n-k+1}, \dots, h_n\rangle.$
Let $\cS_{\text{I}}= \langle g_{c+1}, \dots, g_{n-k}\rangle\subset {\cG}_n$ be the {isotropic} subgroup, which is Abelian.
Let $\cS_{\text{S}}= \langle g_{1}, h_{1}, \dots, g_{c}, h_{c}\rangle \subset {\cG}_n$  be the {symplectic} subgroup, which is non-Abelian.
Since $g_i h_i= -h_i g_i$, $- \mathbb{I}\in\cS_{\text{S}}$. 
Let $\cS=\{g_1 g_2: g_1\in {\cS}_{\text{S}}, g_2\in \cS_{\text{I}}\}$  be the \emph{simplified} stabilizer group, which is non-Abelian.
Since the elements in $\cS_{\text{S}}$ commute with the elements in $\cS_{\text{I}}$, we can safely denote $\cS$  by the notation $\cS_{\text{S}}\times \cS_{\text{I}}$ and similarly in the following.
Then  the minimum distance $d$ of $\cC(\cS')$ is the minimum weight of any element in
$
\overline{\mathcal{S}^{\perp}}\setminus \overline{\mathcal{S}_{\text{I}}}. 
$

Now we prove a Gilbert-Varshamov bound for EA stabilizer codes.
\bt \label{thm:EA GV bound}
There exists an $[[n,k,d;c]]$ EA stabilizer code
provided that
\[
\left(2^{n+k-c}-2^{n-k-c}\right) \sum_{j=0}^{d-1}{n\choose j}3^j \leq 4^{n}-1.
\]
As $n$ becomes large, we have
\[
R= 1+\rho- \delta \log_2 3 -H_2(\delta),
\]
where $H_2(x)=-x\log_2x -(1-x)\log_2(1-x)$ is the binary entropy function. 

\et
\begin{proof}
The first part is similar to the proof of the Gilbert-Varshamov bound for quantum stabilizer codes~\cite{CRSS97,ABKL00b,AKS06a}.
Let $N_k$ be the number of all $[[n,k;c]]$ EA stabilizer codes.
Let $\cM$ denote the multiset
\[
\bigcup_{\cS }\overline{\mathcal{S}^{\perp}}\setminus \overline{\mathcal{S}_{\text{I}}},
\]
where the union is over all simplified stabilizer subgroups  $\cS= \cS_{\text{S}}\times \cS_{\text{I}}$ of $[[n,k;c]]$ EA stabilizer codes.
For each $\cS$, there are $\left(2^{n+k-c}-2^{n-k-c}\right)$ elements in $\overline{\mathcal{S}^{\perp}}\setminus \overline{\mathcal{S}_{\text{I}}}$.
Recall that the $n$-fold Clifford group is the set of unitary operators that preserve the $n$-fold Pauli group $\cG_n$ by conjugation.
It is known that the $n$-fold Clifford group is transitive on ${G}_n\setminus \{\mathbb{I}\}$. That is, for $E,F\in {G}_n\setminus \{\mathbb{I}\}$,
there exists an  $n$-fold  Clifford operator $U$ such that $UEU^{\dag}=F$.
If $E\in \cS$, then $F\in U\cS U^{\dag}$,
where $U\cS U^{\dag}\triangleq\{UgU^{\dag}:g\in \cS  \}$.
Thus each nonidentity element   $E\in \overline{G}_n$ appears
\[
N_k\cdot \frac{2^{n+k-c}-2^{n-k-c}}{4^n-1}
\]
times in $\cM$. Now we delete from $\cM$ those $\overline{\mathcal{S}^{\perp}}\setminus \overline{\mathcal{S}_{\text{I}}}$ with at least one nonidentity element of weight less than $d$.
At most
\[
\sum_{j=0}^{d-1}{n\choose j}3^j\cdot N\cdot \frac{2^{n+k-c}-2^{n-k-c}}{4^n-1}
\]
$\cS$ are removed from the union of $\cM$.  If this number is smaller than $N_k$,  there exists an $[[n,k,d;c]]$ EA stabilizer code.

 The second part is from
 Stirling's approximation:
\begin{align}
\frac{1}{n}\log_2{n\choose d} =H_2\left(\frac{d}{n}\right)+o(1), \label{eq:stirling}
\end{align}
where $o(1)$ tends to $0$ as $n$ increases.
\end{proof}
The Gilbert-Varshamov bound for EA stabilizer codes suggests that entanglement-assisted quantum codes may have higher code dimension  than codes without entanglement assistance for fixed $n$ and $d$.

\section{Upper Bounds for EA Quantum Codes} \label{sec:LP bounds}
In this section, we will derive Delsarte's algebraic linear programming bounds for \emph{general} quantum codes with {entanglement} assistance. 
We first derive the main theorem, similar to \cite[Theorem 4]{AL99} and \cite[Theorem 3]{ALB16}.
\begin{theorem}\label{thm:gen_upper_bound}

Let $\cQ$ be an $((n, M, d;c))$ EA quantum code. Let
\[
f(x,y)= \sum_{i=0}^{n}\sum_{j=0}^{c} f_{i,j}K_i(x;n)K_j(y;c)
\]
be a polynomial
with nonnegative coefficients in the Krawtchouk expansion~(\ref{eq:expansion}), i.e., $f_{i,j}\ge 0$ $\forall i,j.$
Assume that
\begin{align}
f_{i,0}> 0, &\mbox{ for $0\leq i\leq d-1$};\label{eq:polin_cond2}\\
f(x,y)\le 0, &\mbox{ for $x\ge d$ or $y\geq 1$}.  \label{eq:polin_cond1}
\end{align}
Then
\begin{equation*}\label{eq:general_bound}
M\le\frac{1}{2^{n+c}} \max_{0\le l\le d-1} {f(l,0)\over f_{l,0}}.
\end{equation*}
If $\cQ$ is nondegenerate, then
\begin{equation}\label{eq:nondegenerate general_bound}
M\le\frac{1}{2^{n+c}}  {f(0,0)\over f_{0,0}}.
\end{equation}
\end{theorem}
\begin{proof}
Suppose $\cQ$ is an $((n, M, d;c))$ EA quantum code with split weight enumerators $\{B_{i,j}\}$ and $\{B_{i,j}^{\perp}\}$. Then
\begin{align*}
& 2^{n+c}M\sum_{i=0}^{d-1} f_{i,0}B_{i,0}\le  2^{n+c} M \sum_{i=0}^n\sum_{j=0}^{c} f_{i,j} B_{i,j}\\
\stackrel{(a)}{=}&2^{n+c}M \sum_{i=0}^n\sum_{j=0}^{c} f_{i,j}
\cdot {1\over 2^{n+c}M} \sum_{u=0}^{n}\sum_{v=0}^{c} B_{u,v}^\perp K_i(u;n)K_j(v;{c})\\
\stackrel{}{=}&\sum_{u=0}^{n}\sum_{v=0}^{c} B_{u,v}^\perp f(u,v)\\
=&\sum_{i=0}^{d-1} B_{i,0}^\perp f(i,0)+\sum_{i=d}^n B_{i,0}^\perp f(i,0)
+\sum_{i=0}^n\sum_{j=1}^c B_{i,j}^\perp f(i,j)\\
\stackrel{(b)}{\le} &\sum_{i=0}^{d-1} B_{i,0}^\perp f(i,0)=\sum_{i=0}^{d-1} B_{i,0} f(i,0),
\end{align*}
where $(a)$ is by  (\ref{cor:MI_3})  from Theorem~\ref{thm:EA MI}; 
$(b)$ follows from assumption  (\ref{eq:polin_cond1}). The last equality is by (\ref{eq:thm2-2}) from Theorem~\ref{thm:EA MI}.
Thus
$$
M\le  \frac{1}{2^{n+c}} {\sum_{i=0}^{d-1} B_{i,0} f(i,0) \over \sum_{l=0}^{d-1} f_{l,0}B_{l,0}}
\le \frac{1}{2^{n+c}} \max_{0\le l\le d-1} {f(l,0)\over f_{l,0}}.
$$

If $\cQ$ is nondegenerate,  we have its split weight enumerator $B_{i,0}=0$ for $0<i<d$.
Using this additional constraint in the above proof, we have (\ref{eq:nondegenerate general_bound}).

\end{proof}

It remains to find good polynomials $f(x,y)$ that satisfy (\ref{eq:polin_cond2}) and (\ref{eq:polin_cond1}).

\subsection{Singleton  Bound for EA Quantum Codes}
\label{sec:singleton}
The Singleton bound for EA stabilizer codes was proposed in~\cite{BDM06}, which is obtained by an information-theoretical approach~\cite{Pre98b}. 
Herein we prove a Singleton bound for EA quantum codes. Our bound applies to nonadditive EA quantum codes as well, and is thus more general than the one in~\cite{BDM06}.
\begin{theorem}
For  an $((n,M,d;c))$ EA quantum code $\cQ$, if $d\leq (n+2)/2$, then
\begin{align*}
M\leq  2^{n+c-2(d-1)} .
\end{align*}
If $\cQ$ is nondegenerate, the bound holds for any $d$.
\end{theorem}
\begin{proof}

Let
\begin{align}
f(x,y)&=4^{n-d+1} \prod_{i=d}^n \left(1 -{x\over i}\right)\cdot 4^{c}  \prod_{j=1}^c (y-j)\notag\\
&=4^{n+c-d+1}\frac{{n-x\choose n-d+1}}{{n\choose d-1}}\prod_{j=1}^c (y-j).\label{eq:singleton_poly}
\end{align}
From (\ref{eq:expansion}), (\ref{orth}), and (\ref{eq:KP1}), after some calculation, we have $$f_{i,j}= {{n-i\choose d-1}\over{n\choose d-1}}  \geq 0.$$
It can be easily checked  that $f_{i,0}>0$ for $0\leq i\leq d-1$ and
$f(x,y)=0$ if $x\geq d$ or $y\geq 1$. 
Also $f(0,0)/f_{0,0}\geq f(i,0)/f_{i,0}$ for $1\leq i\leq d-1$ if $d\leq (n+2)/2$. Thus by Theorem \ref{thm:gen_upper_bound},
we have
$M\leq \frac{1}{2^{n+c}}\frac{f(0,0)}{f_{0,0}}= 2^{n+c-2(d-1)}.$

\end{proof}

The assumption $d\leq  (n+2)/2$  is reasonable for quantum codes~\cite{AL99} because of the no-cloning theorem~\cite{WZ82}; however, entanglement can increase the error-correcting ability of quantum codes as Grassl has proposed a construction of EA stabilizer codes with $d> (n+2)/2$ \cite{Gra17}.

The argument used in~\cite{BDM06} does not generalize to this case of $d> (n+2)/2$.
Also, the polynomial $f(x,y)$ in (\ref{eq:singleton_poly}) does not work for $d>\frac{n+2}{2}$
since
$f_{l,0}={{n-l\choose d-1}\over{n\choose d-1}} =0 $ if $n-l<d-1$,
which will lead to a trivial upper bound $M<\infty$.
To solve this problem, a possible way is to introduce another polynomial $h(x,y)\triangleq f(x,y)+g(x,y)$
such that $g_{l,0} >0$ for those $l$ with $f_{l,0}=0$.
A candidate is
\[
g^a(x,y)= 4^a \prod_{i=1}^n (x-i)\prod_{j=1}^c (y-j)
\]
with coefficient \[
g^a_{x,y}= 4^{a-n-c}
\]
in the Krawtchouk expansion, where $a$ is some real number chosen appropriately.
Apparently, we can apply  Theorem~\ref{thm:gen_upper_bound} with polynomial $g^a(x,y)$ and obtain another trivial bound $$M\leq 2^{n+c}.$$
Now define
\[
h^a(x,y)=f(x,y)+g^a(x,y)
\]
with coefficient
\[
h^a_{x,y}=f_{x,y}+g^a_{x,y}.
\]
By linearity, we can apply Theorem~\ref{thm:gen_upper_bound} with $h^a(x,y)$. Optimizing over appropriate real numbers $a$, we have the following theorem.
\begin{theorem}(Refined Singleton Bound) \label{thm:ea_singleton}
For  an $((n,M,d;c))$ EA quantum code $\cQ$,  then
\begin{equation*}\label{eq:general_bound}
M\le\frac{1}{2^{n+c}} \min\left\{\max_{0\le l\le d-1} {f(l,0)\over f_{l,0}},\min_{0\leq a\leq n+c} \max_{0\le l\le d-1} {h^a(l,0)\over h^a_{l,0}} \right\}.
\end{equation*}
\end{theorem}
Note that the range $0\leq a\leq n+c$ can be enlarged.

The smallest EA quantum code with $d> (n+2)/2$ in~\cite{Gra17} has parameters $[[9,1,6;1]]$.
Applying this theorem with $n=9,d=6,c=1$, and optimizing over $a$ from $0$ to $10$ with increment $0.001$,
we have $M\leq 3.74.$
Therefore, the $[[9,1,6;1]]$ EA quantum code is optimal.

\textbf{Remark:}
The refined Singleton bound for EA quantum codes in Theorem~\ref{thm:ea_singleton} is not monotonic in $d$, while it appears to be monotonic in $c$.
For $n=9,d=7,c=1$, we have $M\leq 5.63$; for $n=9,d=8,c=1$, we have $M\leq 5.18$.
Thus the quest for a good general bound remains open. It is our future direction to find suitable auxiliary polynomials.

\subsection{Hamming Bounds for EA Quantum Codes}

It is known that a nondegenerate $[[n,k,d;c]]$ EA stabilizer code satisfies the EA Hamming bound~\cite{Bowen02,BDM062}.
Herein we  derive a Hamming bound for general EA quantum codes.

\bt
For  an  unrestricted (nondegenerate or degenerate) $((n,M,d;c))$ EA quantum code $\cQ$,
\begin{align}
M\leq 2^{n+c}\max_{0\le l\le d-1} \frac{\sum_{i=0}^{t}\sum_{j=0}^t\sum_{r=0}^{n-l} \alpha(l,i,j,r)}{ \left(\sum_{i=0}^{t} K_i(l;n)\right)^2}, \label{eq:EA deg Hamming bound}
\end{align}
where  $t= \lfloor\frac{d-1}{2}\rfloor$ and $\alpha(l,i,j,r)$ is defined in (\ref{eq:alpha}).

If $\cQ$ is nondegenerate,
\begin{align}
M\sum_{j=0}^t 3^j{n\choose j}\leq {2^{n+c}}{}. \label{eq:EA Hamming bound}
\end{align}

\et
\begin{proof}

Let $f_{l,j}^H=F_l^H $, where
\begin{align}
F_{l}^H = \left(\sum_{i=0}^{t} K_i(l;n)\right)^2. \label{eq:hamming_polynomial}
\end{align}
Using (\ref{orth}) and (\ref{eq:pipj}),
one can show that
\[
f^H(x,y)= 4^{n+c}\sum_{i=0}^{t}\sum_{j=0}^t\sum_{r=0}^{n-x} \alpha(x,i,j,r)
\]
is the polynomial with coefficients $f_{l,j}^H$ in the Krawtchouk expansion.
It can be checked easily that $f^H(x,y)=0$ if $x\geq d$ or $y\geq 1$. Thus we can apply Theorem~\ref{thm:gen_upper_bound}.

\end{proof}

For  degenerate $((n,M,d;c))$ EA quantum codes,
(\ref{eq:EA deg Hamming bound}) and (\ref{eq:EA Hamming bound}) coincide   if $\max_{0\leq l\leq d-1}{f^H(l,0)\over f^H_{l,0}} $ is achieved at $l=0$.
The region of $n,d$ where the quantum Hamming bound holds for degenerate quantum stabilizer codes has been discussed in~\cite{LX10,ALB16}. 
So far, there is no evidence of degenerate quantum codes that violate the nondegenerate Hamming bound~(\ref{eq:EA Hamming bound}).
The same analysis can be considered here.
Note that for fixed $n, d$,  the value of $l$ that maximizes ${f^H(l,0)\over f^H_{l,0}}$ does not depend on~$c$.
In Fig.~\ref{fig:hamming}, we plot the degenerate and nondegenerate Hamming bounds at $d=9$ and $c=3$.
The two bounds coincide after $n=19$.
\begin{figure}[h]
\vspace{-0.5cm}
\[
\includegraphics[width=9.0cm]{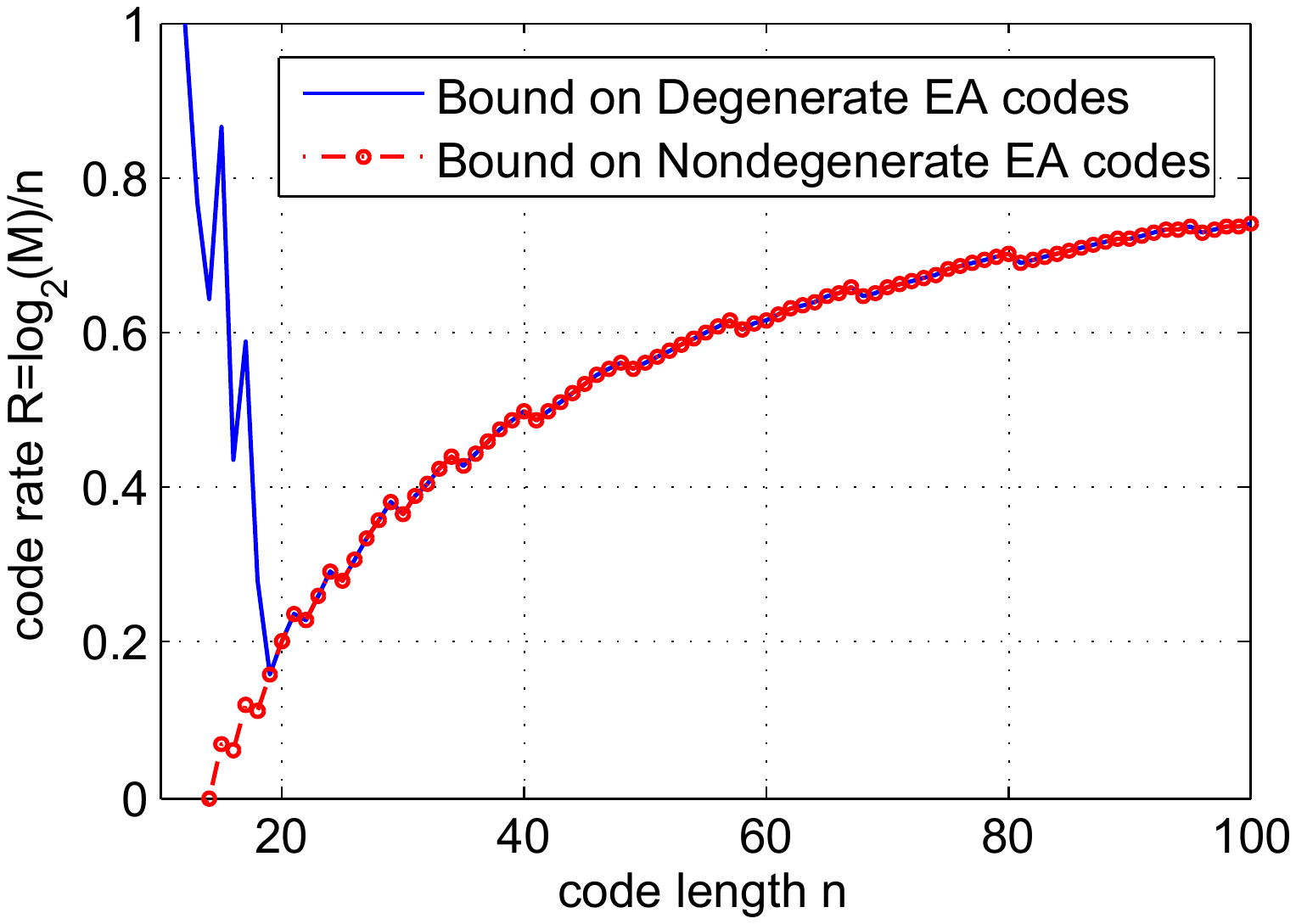}
\]
\vspace{-1.0cm}
  \caption{Plots of the degenerate and nondegenerate Hamming bounds at $d=9$ and $c=3$.
}
\label{fig:hamming}
\end{figure}
We have observed similar behaviors for several values of $d$ and $t$.
Thus we have the following conjecture:
\bcj
The nondegenerate Hamming bound  (\ref{eq:EA Hamming bound}) holds for degenerate $((n,M,d;c))$ EA quantum codes for $n\geq N(t)$,
where $N(t)$ does not depend on $c$.
\ecj
The first few values of $N(t)$ are listed in Table~\ref{tb:EA_hamming bound}.
Similar results have been observed for data-syndrome codes~\cite{ALB16}.

\begin{table}[h]
\[
\begin{tabular}{|c|c|c|c|c|c|c|c|c|c|}
  \hline
  t    & 2 & 3 & 4 & 5 & 6 & 7 & 8 & 9 & 10 \\
  \hline
  N(t) & 9 & 14 & 19 & 24 & 30 & 35 & 40 & 46 & 51 \\
  \hline
\end{tabular}
\]
\caption{Values of $N(t)$ for $2\leq t\leq 10$.}\label{tb:EA_hamming bound}
\end{table}

Recall that Theorem \ref{thm:EA GV bound} suggests that $M$ can be larger by introducing entanglement assistance.
Also the degenerate and nondegenerate Hamming bounds diverge at low code rate from the above discussion.
Consequently,  it is likely that EA quantum codes violate the nondegenerate Hamming bound,
as evidences have been provided in \cite{LGX14}: A family of degenerate $[[n=4t,1,2t+1;1]]$ EA stabilizer codes for $t\geq 2$ has been constructed, which violate (\ref{eq:EA Hamming bound}).

In~\cite{AL99} it is proven that for large $n$ the nondegenerate Hamming bound holds for all (including degenerate) quantum codes.
Using similar argument we can show that the same is true for EA quantum codes. Hence, applying  (\ref{eq:stirling}) to (\ref{eq:EA Hamming bound}), we obtain the following corollary.
\bc
For any $((n,M,d;c))$ EA quantum code with large  $n$,
\[
\frac{\log_2M}{n}\leq 1+\rho-\frac{\delta}{2}\log_2 3- H_2(\frac{\delta}{2})+o(1),
\]
where $\rho=c/n$ and $\delta= {d}/{n}<1/3$.
\ec
By introducing entanglement assistance, distance can be increased~\cite{LB10}.
Again consider the family of $[[n=4t,1,2t+1;1]]$ EA stabilizer codes. These codes have relative distance $\delta=0.5$, which is beyond the working region of the Hamming bound.

\subsection{The First Linear Programming Bound for EA Quantum  Codes} \label{sec:LP1}
The following polynomial is used to prove the first linear programming  bound of classical binary or nonbinary codes~\cite{MRRW77,Aal77,Lev95} and also quantum codes~\cite{AL99}
\begin{align}
F(x)=& \frac{1}{a-x} \left( K_{t+1}(x;n)K_{t}(a;n)-\right.\left.K_{t}(x;n)K_{t+1}(a;n)\right)^2 \label{eq:f(x)}
\end{align}
with
\begin{align}
F_j 
=&\frac{4\cdot 3^t }{t+1}{n \choose t}
\left( K_t(a) \sum_{i=0}^t \frac{K_i(a)}{3^i {n \choose i}}
\sum_{s=0}^{n-j}  \alpha(j,t+1,i,s)\right.  \nonumber \\
&\left.
 -K_{t+1}(a) \sum_{i=0}^t \frac{K_i(a)}{3^i {n \choose i}} \sum_{s=0}^{n-j}
  \alpha(j,t,i,s)\right)
\label{eq:f_x}
\end{align}
for  some suitable $a$.

Herein we will first discuss the LP-1 bound for quantum codes of finite lengths. It is straightforward to generalize the results to EA quantum codes by choosing $f(x,y)=4^c F(x)\prod_{j=1}^c (y-j)$
and
$f_{j,l}= F_j$,
where $F_j$ is defined in (\ref{eq:f_x}).
In particular, we will show that the  LP-1 bound is better than the Hamming bound for large $d/n$.

\subsubsection{LP-1 bound for quantum codes of finite lengths}
\bt Let $t$ be the smallest integer such that $r_t<d$. Then
for  an  unrestricted (nondegenerate or degenerate) $((n,M,d;c))$ EA quantum code $\cQ$,
\begin{align}
M\leq \frac{1}{2^{n-c}}\min_{a\in (r_{t+1},r_t)} \max_{0\le j\le d-1} {F(j)\over F_{j}},
\end{align}
where  $F(j)$ is defined in (\ref{eq:f(x)})  and $F_j$ is defined in (\ref{eq:f_x}).
If $\cQ$ is nondegenerate,
\begin{align}
M\leq \min_{a\in (r_{t+1},r_t)} \frac{1}{2^{n-c}} \frac{\left( 3^{t+1}{n\choose t+1} K_t(a)- 3^t{n\choose t} K_{t+1}(a)\right)^2}{aF_0}.
\end{align}
\et
\begin{proof}
By the Christoffel-Darboux formula (\ref{ch-dar}),
\begin{eqnarray*}
F(x)&=&\frac{4\cdot 3^t}{t+1}{n \choose t}\left\{K_{t+1}(x)K_t(a)-K_t(x)K_{t+1}(a)\right \}\\
&&\times\sum_{i=0}^t \frac{K_i(x)K_i(a)}{ 3^i {n \choose i }} \\
&=& \frac{4\cdot 3^t }{t+1}{n \choose t} \left(K_t(a) \sum_{i=0}^t \frac{K_i(a)}{3^i {n \choose i}} K_{t+1}(x)K_i(x)\right.\\
&&\left.-K_{t+1}(a) \sum_{i=0}^t \frac{K_i(a)}{3^i {n \choose i}} K_{t}(x)K_i(x)\right).
\end{eqnarray*}
Using (\ref{eq:pipj}) we obtain
$$
F(x) \prod_{j=1}^c (y-j)=4^c \sum_{j=0}^n F_j K_j(x),
$$
with
$F_j$ defined in (\ref{eq:f_x}).

From (\ref{eq:f_x}) and (\ref{eq:f(x)}) it follows that by choosing $a\in (r_{t+1},r_{t})$, we guarantee that $F_j\ge 0$ and $F(x)<0,~x\ge d$.
Thus we can apply Theorem~\ref{thm:gen_upper_bound} with $f(x,y)=4^c F(x)\prod_{j=1}^c (y-j)$ and
$f_{j,l}= F_j$.
\end{proof}

Even for small values of $n$ we have to manipulate by very large numbers during computations of $F(j)$ and $F_j$ since
the values of Krawtchouk polynomials and binomial coefficients in (\ref{eq:f_x}) grow very rapidly with $n$.
The following things allow one to speed up computations significantly.

First, the analysis of (\ref{eq:f_x}) shows the limits of summations can be computed more accurately, leading to the equation:
\begin{align}
F_j& 
=\frac{4\cdot 3^t }{t+1}{n \choose t}\nonumber\\
&\times\left( K_t(a) \sum_{i=0}^t \frac{K_i(a)}{3^i {n \choose i}}
\sum_{{s=\max\{0,(t+i+1-2j)/2,}\atop{t+1-j\}}}^{\min\{n-j,(t+i+1-j)/2\}}  \alpha(j,t+1,i,s)  \right.  \nonumber \\&
\left.
 -K_{t+1}(a) \sum_{i=0}^t \frac{K_i(a)}{3^i {n \choose i}} \sum_{{s=\max\{0,(t+i-2j)/2,}\atop{t+1-j\}}}^{\min\{n-j,(t+i-j)/2\}}
 \alpha(j,t,i,s)
\right).
\label{eq:f_x_new}
\end{align}
Second, in (\ref{eq:f_x_new}) the two summations over $s$ do not depend on $a$. So, for given $j$ and $i$, we can pre-compute and reuse them for getting $F_j$ for different values of $a$.

Third, computation of $K_t(x)$ according to (\ref{rec}) is much faster than according to (\ref{kraw}).

To get a good bound for a particular value of $n$, we have to optimize the choice of parameters $t$ and $a$. The following procedure is used:
\begin{enumerate}[1.]
\item For finding the smallest $t$ such that $r_t<d$, we start with $t_0=\frac{3n}{4}-\frac{d}{2}-\frac{1}{2}\sqrt{3d(n-d)}$ and increase $t$ until  $K_t(d)<0$.
\item Find $r_t$ and $r_{t+1}$.
\item Find
\begin{equation}\label{eq:aopt}
a_{opt}=\argmin_{a\in (r_{t+1},r_t)} \max_{0\leq j\leq d-1} \frac{F(j)}{F_j}.
\end{equation}
\end{enumerate}
For example, for $n=50$ and $d=15$, we have $t=15,~r_{t+1}\approx 13.543,~r_t\approx 14.510$, and the behavior of $\max_{0\leq j\leq d-1} F(j)/F_j$ as a function of $a$ is shown in Fig.~\ref{fig:Fj/fj}. In all our computations for different $n$ and $d$ this function happened to be convex. So, we conjecture that this is always the case.
\begin{figure}[h]
\[
\includegraphics[width=8.0cm]{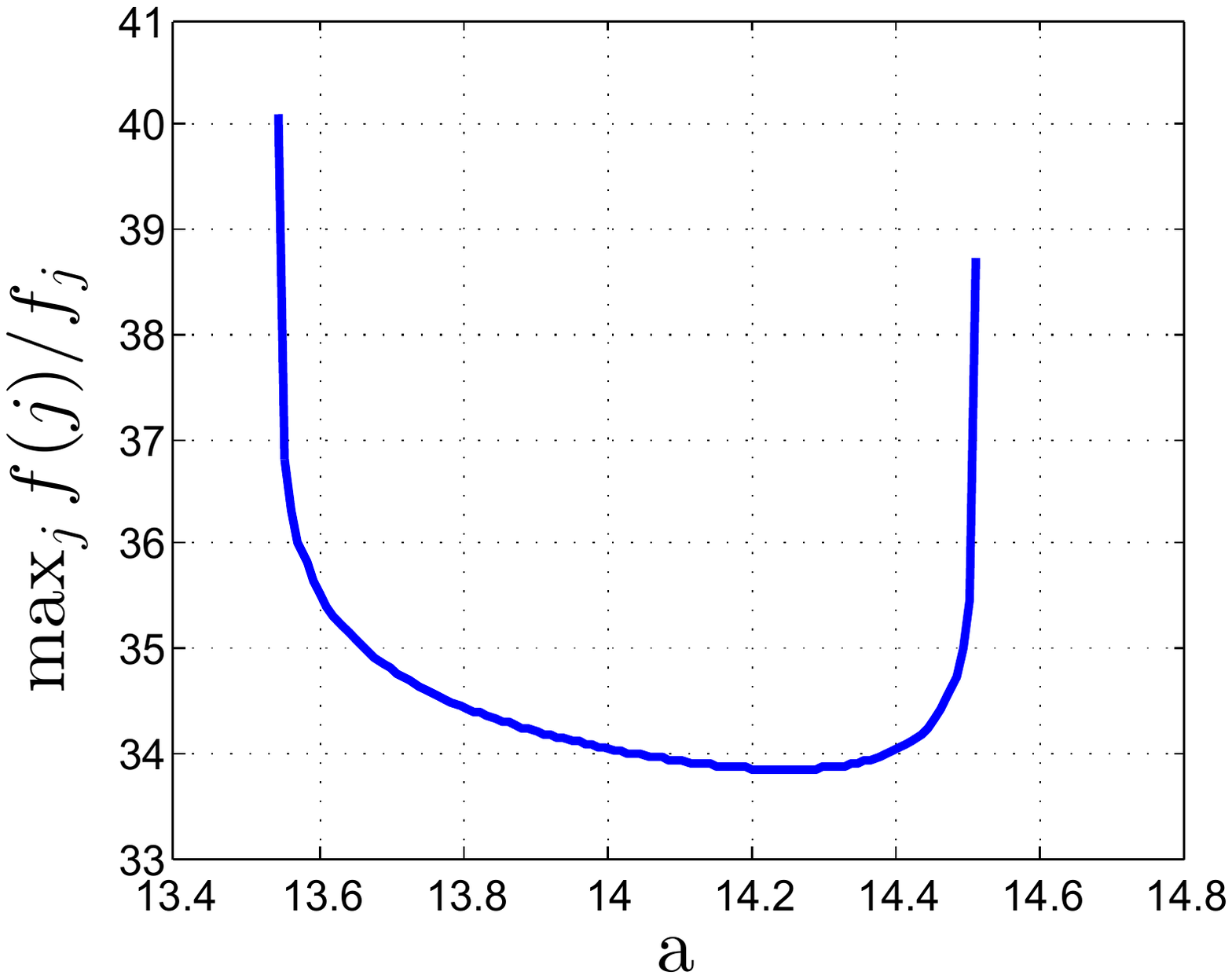}
\]
\vspace{-0.5cm}
\caption{$\max_j F(j)/F_j$ as a function of parameter $a$ for $n=50$ and $d=15$.}
 \label{fig:Fj/fj}
\end{figure}

The analysis of (\ref{eq:f_x}) shows that if we fix $n$ and start increasing $d$, then at a certain moment we will have $F_j=0$ for some $j\le d-1$. For instance, for $n=30$ this happens when $d\ge 15$ and for $n=60$ this happens for $d\ge 25$.  To overcome this problem one may try to choose larger $t$ so that it is not the first $t$ for which $r_t<d$. We leave this possibility, however, for future work. In this work we assume that if $F_j=0, j\le d$, then LP-1 bound is not applicable.

It is shown in \cite{AL99} that in the asymptotic regime, as code length $n$ tends to infinity,  the LP-1 bound beats the Hamming bound. Similarly, the LP-1 bound beats the Hamming bound for codes of finite lengths. In Fig.~\ref{fig:n1000} we plot
the Hamming and LP-1 bounds on the code rate $R$
$$
R=\frac{\log_2 M}{n}\le  \log_2 \left(\max_{0\leq j\leq d-1} \frac{F(j)}{F_j}\right) -1+\rho,
$$
for $\rho=0$  and $n=1000$. One can see that for large values of $d$, the LP-1 bounds visibly improves the Hamming bound.

\begin{figure}[h]
\[
\includegraphics[width=9.0cm]{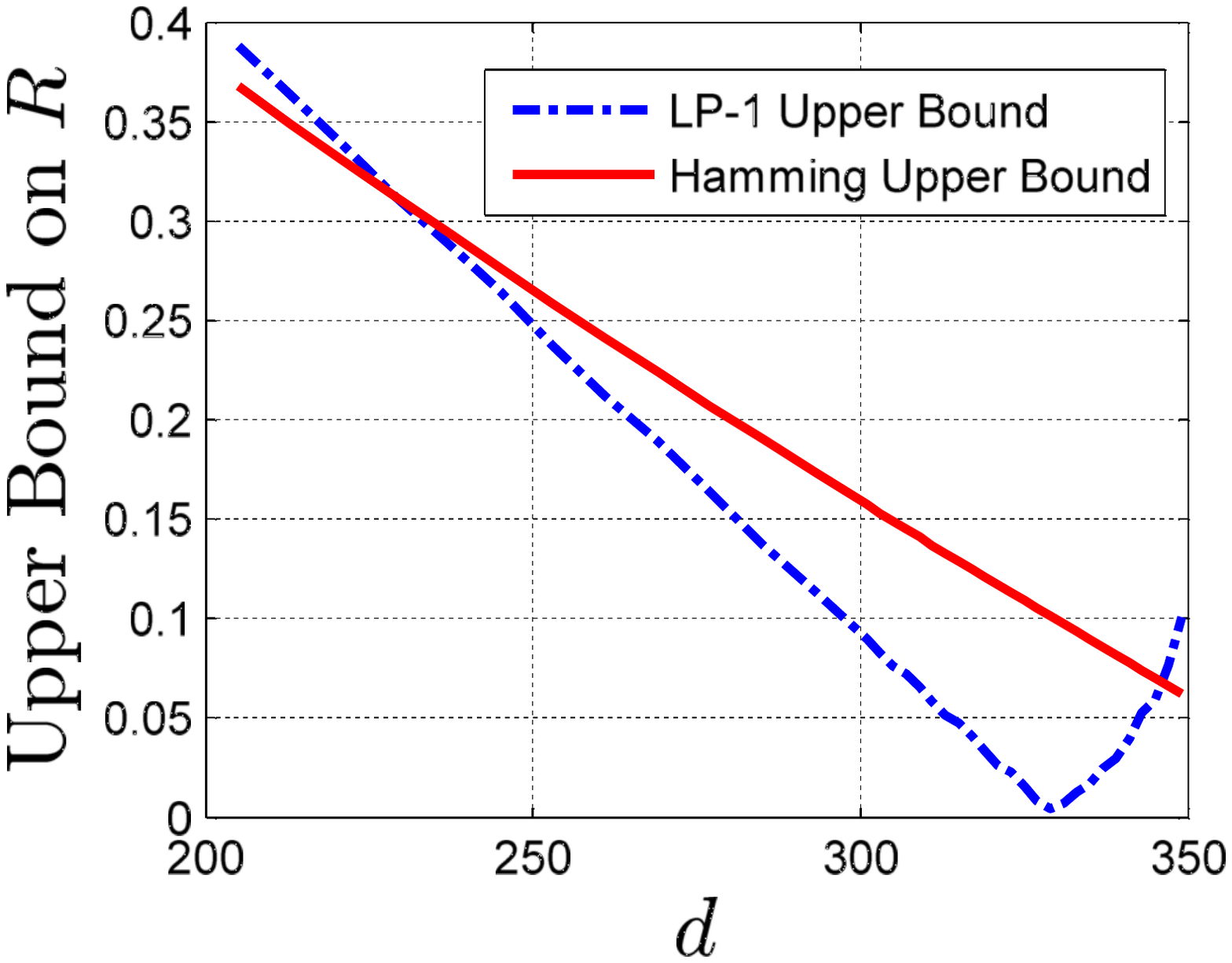}
\]
\caption{The LP-1 and Hamming Upper Bounds for $n=1000, \rho=0$ and $d=200,\ldots,350$.}
 \label{fig:n1000}
\end{figure}

\subsubsection{LP-1 bound for large values of $n$}

Packages like Maple and Mathematica allow one to use very large precisions for computations over real numbers.  For large values of $n$, however, the needed precision becomes overwhelming and infeasible.

To overcome this problem, we compute the {LP-1} bound using only integers as follows. When $n$ grows, the roots $r_{t+1}$ and $r_t$ are getting closer and closer to each other. Numerical computations  show that  choosing $a$ such that $K_t(a)=-K_{t+1}(a)$, we get a value that is very close to $a_{opt}$ defined in (\ref{eq:aopt}). Let us denote this $a$ by $a^*$.  Then we  have
$$
F(j)=\frac{1}{a^*-j}
K_t(a^*)^2\left\{ K_{t+1}(j)+K_{t}(j)\right\}^2.
$$
Next, we lower bound $F_j$ by using only the term with $i=t$ in (\ref{eq:f_x_new}):
\begin{align}
F_j&
\ge F^{(low.bnd.)}_j\nonumber \\
=&K_t(a^*)^2 {4\over t+1}
\left(
\sum_{s=\max\{0,(2t+1-2j)/2,t+1-j\}}^{\min\{n-j,(2t+1-j)/2\}} \alpha(j,t+1,t,s) \right.   \nonumber \\
&
\left.
 +  \sum_{s=\max\{0,(2t-2j)/2,t+1-j\}}^{\min\{n-j,(2t-j)/2\}}
  \alpha(j,t,t,s)\right)\nonumber\\
  \triangleq&K_t(a^*)^2 {4\over t+1} \cdot c_j.
\label{eq:f_x_new2}
\end{align}

Finally, we find a rational number $p/q$ such that $K_{t+1}(p/q)<0,K_{t}(p/q)>0$, and $|K_{t+1}(p/q)|<K_t(p/q)$. This guarantees that $p/q<a^*$.
Hence using $p/q$ in the denominator of (\ref{eq:f(x)}) instead of $a^*$, we make our LP-1 bound larger (worse).

Summarizing these arguments, we obtain that for the optimal LP-1 polynomial $F(j)$ (with $a_{opt}$) we have
\begin{equation}\label{eq:lp1_integers_only}
\frac{F(j)}{F_j}\le \frac{1}{p/q-j}
\left( K_{t+1}(j)+K_{t}(j)\right)^2 \cdot {t+1\over 4}\cdot {1\over c_j},
\end{equation}
where $c_j$ is defined in (\ref{eq:f_x_new2}). During computation of the right-hand side of (\ref{eq:lp1_integers_only}) we need to operate only with integers, no operations over real numbers are needed. We also do not need to operate with real numbers for comparing right hand sides of (\ref{eq:lp1_integers_only}) with different values of $j$.
Thus we do not lose any precision in computations. Using this approach we obtained the results presented in Fig.~\ref{fig:d03} for $\rho=0$.
We consider the case when the code length grows, but the ratio $d/n=0.3$ is fixed. We again observed that starting with some small value of $n$, LP-1 bound visibly improves the Hamming bound.

\begin{figure}[h]
\[
\includegraphics[width=9.0cm]{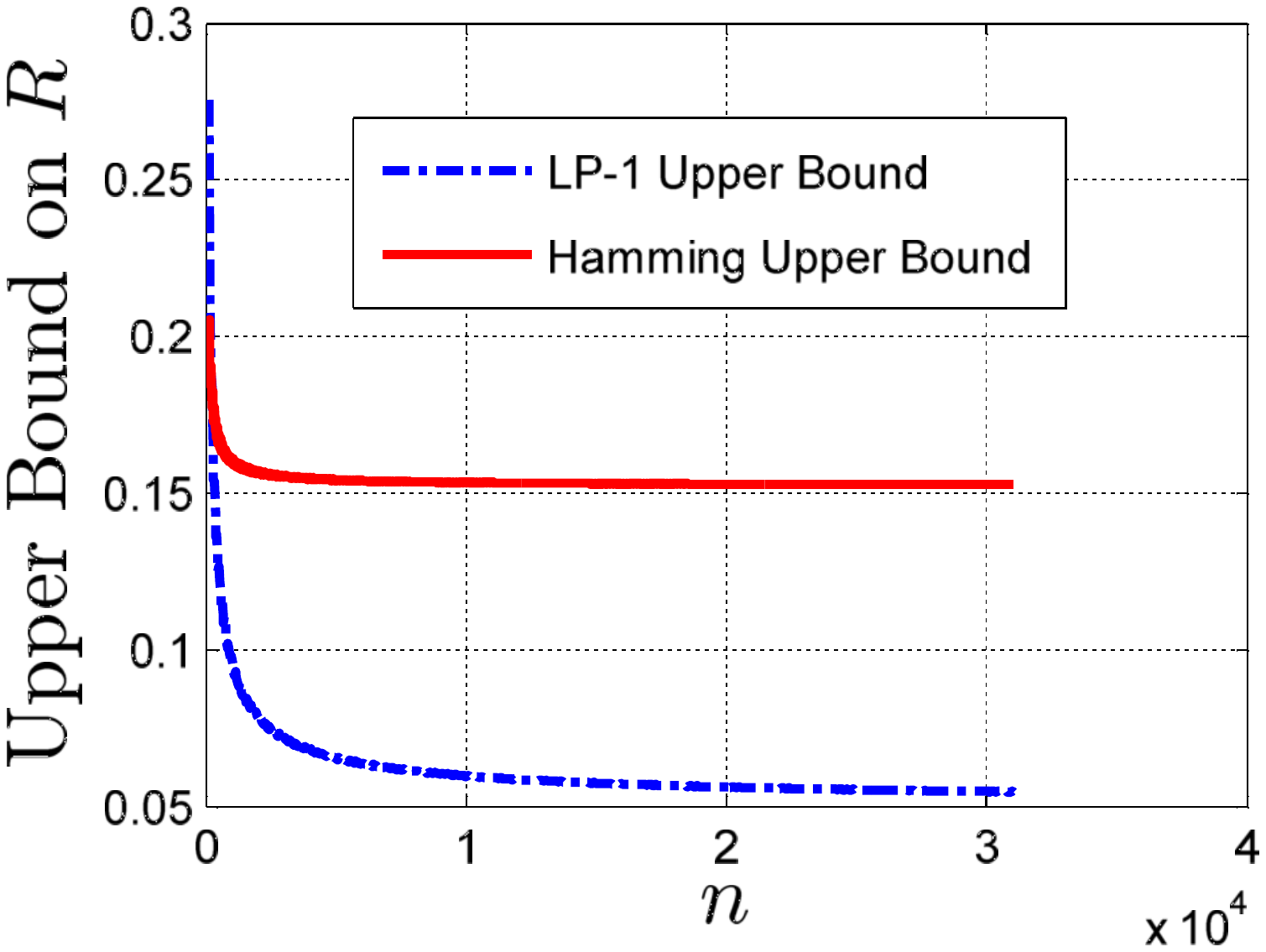}
\]
\caption{LP-1 and Hamming Upper Bounds for code lengths $n=30,\ldots,31010$, $\rho=0$, and $d=0.3n$.}
 \label{fig:d03}
\end{figure}

\subsubsection{Asymptotic case}

Following the same procedure (see \cite{AL99}, for details), we can show that the conventional first linear programming bound holds for EA quantum codes asymptotically.
\bt
For any $((n,M,d;c))$ EA quantum code $\cQ$ with $0\leq \delta\leq 0.3152$,
\begin{align}
R\leq&  \rho- 1+H_2\left(\frac{3}{4}-\frac{\delta}{2}-\frac{1}{2}\sqrt{3\delta(1-\delta)} \right)\nonumber\\
&+ \left(\frac{3}{4}-\frac{\delta}{2}-\frac{1}{2}\sqrt{3\delta(1-\delta)} \right)\log_2 3+o(1) \label{eq:EA LP}
\end{align}
for large $n$. If $\cQ$ is nondegenerate, (\ref{eq:EA LP}) holds for all $\delta$.
\et
Note that this upper bound differs from that for classical quaternary codes by a constant $\left(\rho-1\right)$. 

Let us consider again the family of degenerate $[[4t,1,2t+1;1]]$ EA stabilizer codes.
These codes have asymptotic $R=0$, $\rho=0$, and $\delta=0.5$, which are beyond the working region of the first linear programming bound.

\subsection{Bounds for Maximal-entanglement EA Stabilizer Codes}

In this subsection we consider $[[n,k,d;c=n-k]]$  maximal-entanglement EA stabilizer codes, where 
there is no degeneracy~\cite{LBW13}.
Lai \emph{et al.} proved a Plotkin bound for maximal-entanglement EA stabilizer codes~\cite{LBW10} and its asymptotic version is:
\begin{align*}
\delta \leq 0.75+o(1).
\end{align*}
The asymptotic Gilbert-Varshamov bound, Singleton bound, Hamming bound, first linear programming bound are as follows:
\begin{align*}
R\geq& 1- \frac{\delta}{2} \log_2 3 -\frac{1}{2}H_2(\delta),\\
R\leq& 1- \delta,\\
R\leq& 1-\frac{\delta}{4}\log_2 3- \frac{1}{2}H_2(\frac{\delta}{2})+o(1),\\
R\leq& \frac{1}{2}H_2\left(\frac{3}{4}-\frac{\delta}{2}-\frac{1}{2}\sqrt{3\delta(1-\delta)} \right) \notag\\
&+ \frac{1}{2}\left(\frac{3}{4}-\frac{\delta}{2}-\frac{1}{2}\sqrt{3\delta(1-\delta)} \right)\log_2 3+o(1).
\end{align*}
Thus we have shown that conventional bounds for classical quaternary codes hold for maximal-entanglement EA stabilizer codes.
These bounds are plotted in Fig.~\ref{fig:LP_maximal}.

\begin{figure}[ht]
\hspace{-1cm}\[\includegraphics[width=9.0cm]{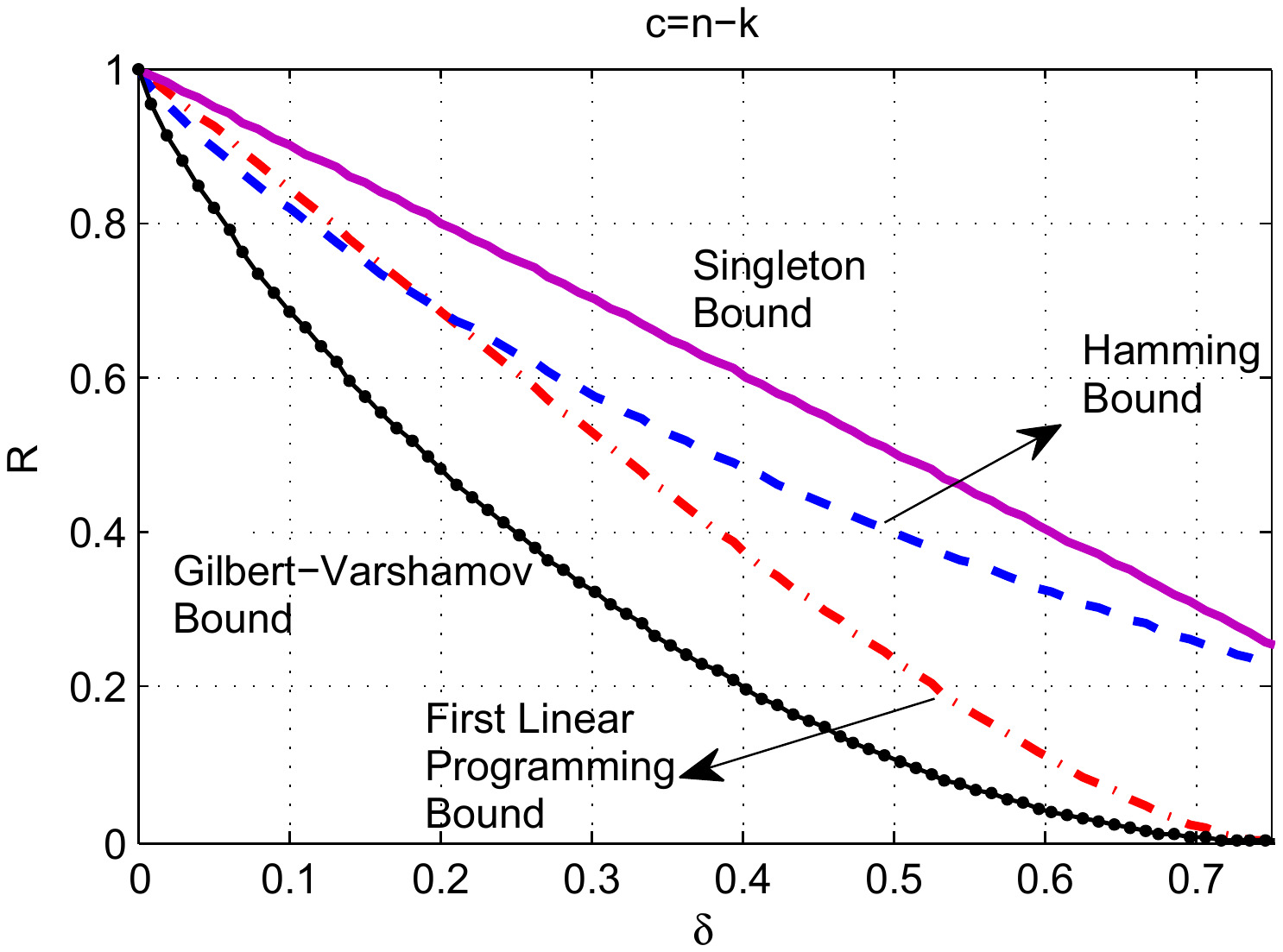}\]
  \caption{Plots of the Singleton bound, Hamming bound, first linear programming bound, and the Gilbert-Varshamov bound. The $x$-axis ranges from 0 to 0.75 because of the Plotkin bound.
}
\label{fig:LP_maximal}
\end{figure}

\section{Upper Bounds for Quantum Stabilizer Codes of Small Length} \label{sec:small codes}
In the previous section we have introduced algebraic linear programming bounds for EA quantum codes. We will discuss linear programming bounds on the minimum distance of EA stabilizer codes of small length.
In particular, we find additional constraints on a stabilizer group and its orthogonal group so that the linear programming bounds can be strengthened.
A table of  upper and lower bounds for any maximal-entanglement EA stabilizer codes with  length $n$ up to $15$ is given in \cite{LBW13}.
However, there are many ``gaps" between the upper  and lower bounds. We will improve that table in this section.

We would like to apply linear programming techniques to upper bound the minimum distance of EA stabilizer codes.
However, since ${\cS^{\perp}}\neq \cS_{\text{I}}$, we cannot use a MacWilliams identity to connect the weight distributions of these two groups.
Observe that \begin{align}
\overline{(\cS_{\text{S}}\times\cS_{\text{I}})^{\perp}}= \overline{\cL\times \cS_{\text{I}}}   \label{eq:EA stabilizerC_duality1}
\end{align}
and
\begin{align}
\overline{\cS_{\text{I}}^{\perp}} = \overline{\cL\times \cS_{\text{S}}\times \cS_{\text{I}} }. \label{eq:EA stabilizerC_duality2}
\end{align}
In \cite{LBW10,LBW13}, these two dualities are considered to obtain the
 linear programming bounds on the minimum distance of small EA stabilizer codes.

Interestingly, the split weight enumerators $\{B_{i,j}\}$ and $\{B_{i,j}^{\perp}\}$ of EA quantum codes defined in (\ref{eq:Bij}), (\ref{eq:Bij_perp})
have enough information for our linear programming problem.
Let
$\cS_{\text{S}}'= \langle g_{1}', h_{1}', \dots, g_{c}', h_{c}'\rangle,$ 
$\cS_{\text{I}}'= \langle g_{c+1}', \dots, g_{n-k}'\rangle,$ 
and $\cL'= \{g^A\otimes \mathbb{I}^B: g\in \cL\}$. Define
\begin{align}
A_{i,j}=& | \{E_1\otimes E_2\in \overline{\cS_{\text{S}}'\times \cS_{\text{I}}'}: \notag\\
&E_1\in \overline{\cG}_n, \wt{E_1}=i,  E_2\in \overline{\cG}_c, \wt{E_2}=j \}|, \label{eq:A_ij}\\
A_{i,j}^{\perp}=& | \{E_1\otimes E_2\in \overline{\cL'\times  \cS_{\text{S}}'\times  \cS_{\text{I}}'}:\notag\\
 &E_1\in\overline{\cG}_n, \wt{E_1}=i, E_2\in \overline{\cG}_c, \wt{E_2}=j \}|. \label{eq:A_ij_perp}
\end{align}
Observe that $B_{i,j}= A_{i,j}$ and  $B_{i,j}^{\perp}= A_{i,j}^{\perp}$ in the case of EA stabilizer codes. Thus we have the following corollary.
\bc Suppose $\cQ$ is an $[[n,k,d;c]]$ EA stabilizer code with stabilizer group $\cS_{\text{S}}'\times \cS_{\text{I}}'$.
Let  $\{A_{i,j}\}$ and $\{A_{i,j}^{\perp}\}$ be the split weight distributions of $\overline{\cS_{\text{S}}'\times \cS_{\text{I}}'}$ and $\overline{(\cS_{\text{S}}'\times \cS_{\text{I}}')^{\perp}}$, respectively. Then
\begin{equation*}
A_{i,j}^\perp={1\over {2^{n+c-k}}}\sum_{u=0}^n\sum_{v=0}^{c}A_{u,v} K_i(u;n) K_j(v;c).
\end{equation*} \label{cor:MI_EA stabilizer}
\ec
Furthermore, the weight distribution of $\cS_{\text{I}}$ is $\{A_{i,0}\}$
and the weight distribution of $\overline{\cS^{\perp}}$ is $\{A_{i,0}^{\perp}\}$.
Since the minimum distance of $\cC(\cS')$ is $d$, we have
\begin{align}
&A_{i,0}=A_{i,0}^{\perp}, \ i=1,\dots,d-1;
A_{d,0}^{\perp}>A_{d,0}.  \label{eq:Aij distance}
\end{align}
Now we can build a linear program with variables $A_{i,j}, A_{i,j}^{\perp}$ and constraints from Corollary~\ref{cor:MI_EA stabilizer} and (\ref{eq:Aij distance}).
The least $d$ such that this linear program has no feasible solution is an upper bound on the minimum distance for given $n,k,c$,
and this is called the linear programming bound  on the minimum distance   of $[[n,k;c]]$ EA stabilizer codes.
The resulting linear programming bound is at least as tight as the one from (\ref{eq:EA stabilizerC_duality1}) and (\ref{eq:EA stabilizerC_duality2}) used in \cite{LBW10,LBW13},
since this linear program has more constraints and variables.

\subsection{Additional Constraints on the Weight Enumerator Associated with an Non-Abelian Pauli Subgroup} \label{sec:new constraints}

Rains introduced the idea of \emph{shadow enumerators} to obtain additional constraints in the linear program of quantum codes~\cite{Rains96}.
The shadow ${S}h(\cV)$ of a group $\cV (\subseteq \cG_n)$ is the set
\[
\{ E\in \overline{\cG}_n : \langle E,g\rangle_{{\cG}_n}= \wt{g} \mod 2, \ \forall g\in \cV  \}.
\]
If $\cV$ is Abelian, Rains showed that
\begin{align}
W_{Sh(\cV)}(x,y)=\frac{1}{|\mathcal{V}|}W_\cV(x+3y, y-x),\label{eq:MI_shadow}
\end{align}
where $W_\cV(x,y)$ and $W_{Sh(\cV)}(x,y)$ are the weight enumerators of $\cV$ and ${S}h(\cV)$, respectively \cite{Rains96}.
For an EA stabilizer code, the isotropic subgroup $\cS_{\text{I}}$ is Abelian. Thus
\begin{align*}
W_{Sh(\cS_{\text{I}})}(x,y)=\frac{1}{|\cS_{\text{I}}|}W_{\cS_{\text{I}}}(x+3y, y-x).
\end{align*}
Recall that $\{A_{i,0}\}$ is the weight distribution of $\cS_{\text{I}}$. This implies
\begin{align}
\sum_{w'=0}^n (-1)^{w'} A_{w',0} K_{w}(w';n)\geq 0, \ i=0,\dots,n. \label{eq:shadow_enumerator}
\end{align}
However, the situation is more complicated when $\cV$ is not Abelian,
which is the case of the simplified stabilizer group $\mathcal{S}= \cS_{\text{S}}\times \cS_{\text{I}}$ and its orthogonal group.
Herein we derive additional constraints for an non-Abelian Pauli subgroup.

Suppose $\cV\subset {\cG}_n$ is non-Abelian and $\cV$ can be decomposed as $\cV= \cV_{\text{S}}\times \cV_{\text{I}}$, where $\cV_{\text{S}}= \langle g_1, h_1, \dots, g_c,  h_c\rangle$
and $\cV_{\text{I}}= \langle g_{c+1}, \dots, g_{c+r} \rangle$ for some $c$ such that the commutation relations (\ref{eq:commutation_1})-(\ref{eq:commutation_4}) hold~\cite{FCYBC04}.
We categorize different $\cV$'s into the following three types:
\begin{enumerate}[I.]
  \item All the generators of $\cV$ are of even weight.

  \item  $\wt{g_{c+1}}$ is odd  and all the other generators of $\cV_{\text{I}}$ and $\cV_{\text{S}}$ are of even weight.

  \item  $\wt{g_{1}}$ and $\wt{h_{1}}$ are odd and all the other generators of $\cV_{\text{S}}$ and $\cV_{\text{I}}$ are of even weight.

\end{enumerate}

For convenience, we have the following lemma, which can be easily verified.
\bl \label{lemma:weight}
For $g, h\in {\cG}_n$,
\[
\wt{gh}\equiv \wt{g}+\wt{h}+\langle g,h\rangle_{{\cG}_n} \mod 2.
\]
\el
\noindent This lemma shows that the weight of the product of two operators depends on their inner product.

\bl \label{lemma:cases}
Suppose $\cV= \cV_{\text{S}}\times \cV_{\text{I}}$, where $\cV_{\text{S}}= \langle g_1, h_1, \dots, g_c,  h_c\rangle$
and $\cV_{\text{I}}= \langle g_{c+1}, \dots, g_{c+r} \rangle$ for some $c$ such that (\ref{eq:commutation_1})-(\ref{eq:commutation_4}) hold.
Then the generators are of type I, II, or III.
\el
\begin{proof}

Suppose  $\cV_{\text{I}}$ contains some elements of odd weight, say $g_{c+1}$ without loss of generality.
If $g_j$ or $h_j$ is of odd weight for $j\neq c+1$, it is replaced by $g_jg_{c+1}$ or $h_{j}g_{c+1}$, which is of even weight by
 (\ref{eq:commutation_1}), (\ref{eq:commutation_3}), and Lemma \ref{lemma:weight}.
Eqs. (\ref{eq:commutation_1})-(\ref{eq:commutation_4}) hold for the new set of generators. This is type II.

Next, suppose $\cV_{\text{I}}$ contains no elements of odd weight.
Consider the generators of $\cV_{\text{S}}$.
We know that $\cV_{\text{S}}$ has some elements of odd weight, because  one of $g_1$, $h_1$, and $g_1h_1$ must be of odd weight.
Assume $g_{1}$  is of odd weight without loss of generality.
If $g_{i}$ is of odd weight for $i=2, \dots, c$,
we replace it with $g_ig_{1}$, which is of even weight by
(\ref{eq:commutation_1}) and Lemma~\ref{lemma:weight}.
Notice that $h_{1}$ has to be replaced by $h_{1}h_i$ at the same time to maintain the commutation relation in (\ref{eq:commutation_3}).
Thus $\wt{g_{1}}$ is odd and $\wt{g_{2}}, \dots, \wt{g_{c}}$ are even.
If  $h_{i}$ has odd weight for $i=2, \dots, c$,
we replace it with $h_ig_{1}$, which is of even weight by  (\ref{eq:commutation_2}) and Lemma \ref{lemma:weight}.
Similarly, $h_{1}$ has to be replaced by $h_{1}g_i$ to maintain the commutation relation in  (\ref{eq:commutation_2}).
Therefore, we can assume $h_{2},\dots, h_{c}$ are of even weight.

It remains to check whether $h_{1}$ is of odd weight or not.
If $\wt{h_{1}}$ is odd, this is type III.
If $\wt{h_{1}}$ is even, we replace $g_{1}$ with $g_{1}h_{1}$, which is of even weight by  (\ref{eq:commutation_4}) and Lemma~\ref{lemma:weight}.
This case is type I.

\end{proof}%

Suppose $\cV\subset \cG_n $ is generated by $c$ pairs of symplectic partners and $r$ isotropic generators of type~$i$.
Let $N^{\even}_{i}(c,r)$ and  $N^{\odd}_{i}(c,r)$ be the number of elements in $\overline{\cV}$ of  even and  odd weight, respectively.
Below we derive formulas for these two numbers for the three types of generators in Lemma \ref{lemma:cases}.
\bt \label{thm:new constraints}
For $c,r\geq 0$,
\begin{align*}
N^{\even}_{\mbox{I}}(c,r)&=2^{r-1}(4^{c}+2^{c}),\\
N^{\odd}_{\mbox{I}}(c,r)&=2^{r-1}(4^{c}-2^{c}),\\
N^{\even}_{\mbox{III}}(c,r)&=2^{r-1}(4^{c}-2^{c}),\\
N^{\odd}_{\mbox{III}}(c,r)&=2^{r-1}(4^{c}+2^{c}).
\end{align*}
For $c\geq 0$ and $r>0$,
\begin{align*}
&N^{\even}_{\mbox{II}}(c,r)=N^{\odd}_{\mbox{II}}(c,r)=2^{2c+r-1}.
\end{align*}
\et
\begin{proof}
These formulas can be derived by recursion.
We first consider the case  of type I.
It is obvious that $N^{\even}_{\mbox{I}}(0,0)=1$ and $N^{\odd}_{\mbox{I}}(0,0)=0$.
(In this case, $\mathcal{V}$ is the trivial subgroup.)
For $c=1$, we have two generators $g_1$ and $h_1$ of even weight
such that $\langle g_1,h_1\rangle_{{\cG}_n}=1$.
Thus $\wt{g_1 h_1}\equiv \wt{g_1}+\wt{h_1}+ \langle g_1,h_1\rangle_{{\cG}_n} \equiv 1 \mod 2$.
We have  $N^{\even}_{\mbox{I}}(1,0)=3$ and $N^{\odd}_{\mbox{I}}(1,0)=1$.
Since $\mathcal{V}_{\text{I}} \subseteq \mathcal{V}_{\text{S}}^{\perp}$ and the generators of $\mathcal{V}_{\text{I}}$ are of even weight (hence every element of $\mathcal{V}_{\text{I}}$ is of even weight), it follows that
 $N^{\even}_{\mbox{I}}(1,r)=3\cdot 2^r$ and $N^{\odd}_{\mbox{I}}(1,r)=2^r$.

Let $S_c$ denote a group generated by $c$ pairs of symplectic partners.
Now we add a new pair of symplectic partners $g'$ and $h'$ of even weight
and  $g', h' \in S_c^{\perp}$.
An element $E$ in the new group $S_{c+1}$ is of even weight if one of the following cases hold:
 $E\in S_{c}$ and $\wt{E}\equiv 0 \mod 2$ ;
$E=g' E'$ for some  $E'\in S_c$ and $\wt{E'}\equiv 0 \mod 2$ ;
$E=h' E'$ for some  $E'\in S_c$ and $\wt{E'}\equiv 0 \mod 2$ ;
or $E=g'h' E'$ for some  $E'\in S_c$ and $\wt{E'}\equiv 1 \mod 2$.
Therefore
\begin{align*}
&N^{\even}_{\mbox{I}}(c+1,r)= 3 N^{\even}_{\mbox{I}}(c,r)+N^{\odd}_{\mbox{I}}(c,r).
\end{align*}
Similarly, we have
\begin{align*}
&N^{\odd}_{\mbox{I}}(c+1,r)= 3 N^{\odd}_{\mbox{I}}(c,r)+N^{\even}_{\mbox{I}}(c,r).
\end{align*}
Solving this system of two recursive equations, we have
\begin{align*}
&N^{\even}_{\mbox{I}}(c,r)=2^{r-1}(4^{c}+2^{c}),\\
&N^{\odd}_{\mbox{I}}(c,r)=2^{r-1}(4^{c}-2^{c}).
\end{align*}
We can find the formula for the generators of type III along the same lines. 

Now we consider the case of type II for $r>0$.
Let $S_c$ denote a group generated by $c$ pairs of symplectic partners of even weight.
From above, we have $N^{\even}_{\mbox{I}}(c,0)$ and $N^{\odd}_{\mbox{I}}(c,0)$.
Now we add a new generator $g \in S_c^{\perp}$ and $\wt{g}$ is odd.
An element $E$ in the new group $g\times S_{c}$ is of even weight if one of the following two cases hold:
 $E\in S_{c}$ and $\wt{E}\equiv 0 \mod 2$;
or $E=g E'$ for some  $E'\in S_c$ and $\wt{E'}\equiv 1 \mod 2$.
Thus
\begin{align*}
&N^{\even}_{\mbox{II}}(c,1)=N^{\even}_{\mbox{II}}(c,0)+N^{\odd}_{\mbox{II}}(c,0).
\end{align*}
Similarly,
\begin{align*}
&N^{\odd}_{\mbox{II}}(c,1)=N^{\even}_{\mbox{II}}(c,0)+N^{\odd}_{\mbox{II}}(c,0)=N^{\even}_{\mbox{II}}(c,1).
\end{align*}
It implies $N^{\odd}_{\mbox{II}}(c,1)=N^{\odd}_{\mbox{II}}(c,1)=2^c$.
Therefore, we have
\begin{align*}
&N^{\even}_{\mbox{II}}(c,r)=N^{\odd}_{\mbox{II}}(c,r)=2^{2c+r-1}.
\end{align*}

\end{proof}
These formulas serve as constraints on the weight distributions of subgroups that are Abelian or non-Abelian.
Note that $\left|N^{\even}_i(c,0)-N^{\odd}_i(c,0)\right|=2^c$ for $i=\mbox{I, III}$.

In the case of standard stabilizer codes, we have Abelian stabilizer group $\cS$ and  non-Abelian normalizer group $\cS^{\perp}$.
Thus we have the following corollary.
\bc \label{cor:cns stabilizer}
Suppose $\cS$  is an Abelian subgroup of ${\cG}_n$ with $n-k$ generators.
Let $\{A_j\}$ and $\{A_j^{\perp}\}$ be the weight distributions of $\cS$ and $\overline{\cS^{\perp}}$, respectively. Then one of the following three cases holds:
\begin{align*}
   &1.\ \sum_{j: j \text{ even}} A_j =2^{n-k-1}, &\sum_{j: j\text{ even}} A_j^{\perp} &=2^{n+k-1};\\
   &2.\ \sum_{j: j\text{ even}} A_j =2^{n-k}, &\sum_{j: j\text{ even}} A_j^{\perp} &=2^{n-k-1}\left(4^k+2^k\right);\\
   &3.\ \sum_{j: j\text{ even}} A_j =2^{n-k}, &\sum_{j: j\text{ even}} A_j^{\perp} &=2^{n-k-1}\left(4^k-2^k\right).
\end{align*}
\ec

\subsection{Linear Program for EA Stabilizer Codes}
Now we provide a linear program for EA stabilizer codes in this subsection.
Let $n,k,d,c$ be integers. Define  integer variables
\begin{align*}
A_{i,j},A_{i,j}^{\perp}, i=0,\dots, n, j=0,\dots, c. 
\end{align*}
The following constraints are from Theorem~\ref{thm:EA MI}, Corollary~\ref{cor:MI_EA stabilizer}, Eq. (\ref{eq:shadow_enumerator}), and Theorem~\ref{thm:new constraints}:
\begin{align*}
&A_{0,0}=A_{0,0}^{\perp}=1, \ A_{i,j}^{\perp}\geq A_{i,j}\geq 0;\\
&A_{i,0}=A_{i,0}^{\perp}, \ i=1,\dots,d-1;
A_{0,j}=A_{0,j}^{\perp}=0, \ j=1,\dots,c;\\
&\sum_{w=0}^n A_{w,0}= 2^{n-k-c},\ \sum_{w=0}^n\sum_{w'=0}^c A_{w,w'}=2^{n-k+c},\\
&\sum_{w=0}^n A_{w,0}^{\perp} =2^{n+k-c}, \ \sum_{w=0}^n\sum_{w'=0}^c A_{w,w'}^{\perp} =2^{n+k+c};\\
&A_{i,j}^\perp={1\over {2^{n+c-k}}}\sum_{u=0}^n\sum_{v=0}^{c}A_{u,v} K_i(u;n) K_j(v;c);\\
&\sum_{w'=0}^n (-1)^{w'} A_{w',0} K_{w}(w';n)\geq 0, \ i=0,\dots,n;\\
&\left( \sum_{w: \even} \sum_{w'} A_{w,w'} = N^{\even}_{\mbox{i}}(c,n-k-c),\ \mbox{ for $i=$ I or III},\right.\\
& \sum_{ w: \even} A_{w,0}^{\perp} = N^{\even}_{\mbox{i}}(k,n-k-c),
\ \mbox{ for $i=$ I or III}, \\
& \sum_{w: \even}  \sum_{w'}A_{w,w'}^{\perp} = N^{\even}_{\mbox{i}}(k+c,n-k-c),
\ \mbox{ for $i=$ I or III},\\
&\mbox{and } \left.\sum_{ w: \even} A_{w,0} =2^{n-k-c} \right) \\
& \mbox{or } \left(\sum_{w: \even} \sum_{w'}A_{w,w'} = N^{\even}_{\mbox{II}}(c,n-k-c),\right.\\
& \sum_{ w: \even} A_{w,0}^{\perp} = N^{\even}_{\mbox{II}}(k,n-k-c),\\
& \sum_{w: \even}  \sum_{w'}A_{w,w'}^{\perp} = N^{\even}_{\mbox{II}}(k+c,n-k-c),\\
&\mbox{and } \left.\sum_{ w: \even} A_{w,0} = 2^{n-k-c-1}\right).
\end{align*}
If there is no solution to the integer program with variables $A_{i,j},A_{i,j}^{\perp}$ and the above constraints for given $n,$ $k,$ $d,$ and~$c$, then  there is no  $[[n,k,d;c]]$ EA stabilizer code.

We can introduce additional constraints from MacWilliams identities for \emph{general} split weight enumerators~\cite{LHL14}.
In the following, we use the notation in~\cite{LHL14}. Let $\cH_{\overline{\cG}_1}$ be the complex Hilbert space with an orthonormal basis $\ket{I}, \ket{X}, \ket{Y}, \ket{Z}$.
Let $\cH_{\overline{\cG}_n}= \cH_{\overline{\cG}_1}^{\otimes n}$. Then $\overline{\cS}\subset \overline{\cG}_n$ has an exact weight generator $$g^{E}_{\overline{\cS}}=\sum_{g\in\overline{\cS}}\ket{g}\in \cH_{\overline{\cG}_n}.$$
The MacWilliams identity says that
\begin{align}
g^{E}_{\overline{\cS^{\perp}}}=\frac{1}{|\cS|} \cF^{\otimes n} g^{E}_{\overline{\cS}},   \label{eq:MI_exact}
\end{align}
where $\cF$ is the Fourier transform  operator on $\cH_{\overline{\cG}_1}$ and its matrix representation in the ordered basis $\ket{I},\ket{X},\ket{Y},\ket{Z}$ is
\begin{align*}
\cF=\left[\begin{array}{rrrr}
1&1&1&1\\ 1&1&-1&-1 \\ 1&-1&1&-1 \\1&-1&-1&1
\end{array}\right]. 
\end{align*}
 Let $\gamma: \cH_{\overline{\cG}_n}\rightarrow \mathbb{C}$ be a linear functional. Then applying $\gamma$ to (\ref{eq:MI_exact}), we have the MacWilliams identities for general split weight enumerators defined by $\gamma$. (See~\cite{LHL14} for more details.)
For example,  let
\begin{align*}
\gamma_\text{H}(x,y)= x\bra{I}+\sum_{\alpha\in\{X,Y,Z\} } y\bra{\alpha},
\end{align*}
where $x,y\in \mathbb{C}$.
Then the general Hamming weight enumerator of $\cS$ and  its MacWilliams identity are defined by $\gamma^{\otimes n}_\text{H}(x,y)$.
The split weight enumerators (\ref{eq:A_ij}), (\ref{eq:A_ij_perp}) are defined by $\gamma^{\otimes n}_\text{H}(x,y)\otimes \gamma^{\otimes c}_\text{H}(u,v)$.

In general, an arbitrary $\gamma$ and its induced MacWilliams identity will lead to more constraints on the coefficients of $\cS$ and $\cS^{\perp}$.
However, these general split weight enumerators usually introduce too many variables to solve in a computer program. If we know more about a stabilizer group, we can design an effective $\gamma$. In~\cite{CRSS98}, the idea of \emph{refined weight enumerator} is introduced when $\cS$ is known to have an element of weight $u$.
Let \begin{align*}
\gamma_\text{re}(y_0,y_1,y_2)= y_0\bra{I}+y_1\bra{X}+\sum_{\alpha\in\{Y,Z\} } y_2\bra{\alpha}.
\end{align*}
Then a refined weight enumerator is defined as
\begin{align}
\gamma_\text{H}^{\otimes n-u}(x_0,x_1)&\otimes \gamma_\text{re}^{\otimes u}(y_0,y_1,y_2) g_\cS^E \notag\\
&=\sum_{i,j,l} A_{i,j,l} x_0^{n-u-i}x_1^i y_0^{u-j-l}y_1^jy_2^l.\label{eq:refined}
\end{align}
Suppose $\{A_i\}$ is the weight distribution of $\cS$. Then
\[
A_w= \sum_{i,j,l: i+j+l= w}A_{i,j,l}.
\]
Without loss of generality, we may assume $\cS$  has the element $I^{\otimes n-u}\otimes X^{\otimes u}$ of weight $u$.
Since $\cS$ is Abelian, we must have $A_{i,j,l}=0$ for $l\neq 0 \mod 2$. Similar constraints can be found for the distribution $A^{\perp}_{i,j,l}$ of $\overline{\cS^{\perp}}$.
From (\ref{eq:MI_exact}), we have additional constraints from the MacWilliams identities on $A_{i,j,l}$ and $A_{i,j,l}^{\perp}$.

Next we use Mathematica~\cite{Mathematica} to see if the integer program has a feasible solution for some parameters. Since Mathematica is a symbolic manipulation program, the results are reliable.
\be (Nonexistence of [[27,15,5]] quantum codes.) \label{ex:27}
In the case of $c=0$,  $\cS_{\text{S}}$ is trivial and it reduces to the case of  standard stabilizer codes.
Solving  the linear program for stabilizer codes~\cite{CRSS98} with additional constraints from Corollary~\ref{cor:cns stabilizer},
we found that a $[[27,15,5]]$ stabilizer code, if it exists,  must have  $\cS^{\perp}$  generated by type III generators and the only possible weight distribution of $\cS$ is
\begin{align*}
A_0=1, A_{16}=81, A_{18}=1800, A_{22}=1944, A_{24}=270,
\end{align*}
and $A_w=0$, otherwise.

Using  refined weight enumerators with respect to $u=24$,
we have  additional constraints from the MacWilliams identities for the refined weight enumerators,
which produce a linear program with no feasible solution.
Thus we improved the quantum code table at $n=27,k=15$~\cite{Grassl}.

\ee

\be (Nonexistence of [[28,14,6]] quantum codes.)
We found that  a $[[28,14,6]]$ stabilizer code, if it exists,  must have $\cS^{\perp}$ generated by type I generators
and the only possible weight distribution of $\cS$ is
\begin{align*}
&A_0=1, A_{16}=189, A_{18}=5040, \\
&A_{22}=9072,A_{24}=1890,A_{28}=192,
\end{align*}
 and $A_w=0$, otherwise. This implies that a $[[28,14,6]]$ code must be nondegenerate.
By the propagation rule  that the existence of a nondegenerate $[[n,k,d]]$ code implies the existence of an $[[n-1,k+1,d-1]]$ code~\cite[Theorem 6]{CRSS98},
if we have a nondegenerate $[[28,14,6]]$ code, there is a $[[27,15,5]]$ code. However, the existence of a $[[27,15,5]]$ code has been excluded in Example~\ref{ex:27}.
 Thus there is no  $[[28,14,6]]$ code.
 \ee
 \be
Similarly, a $[[23,1,9]]$ stabilizer code, if it exists,  must have $\cS^{\perp}$ generated by type III generators
and the only possible weight distribution of $\cS$ is
\begin{align*}
&A_0=1, A_{10}=10626, A_{12}=78246, A_{14} =478170, \\
& A_{16}=1245519, A_{18}= 1562022, A_{20}= 733194, A_{22}= 86526
\end{align*}
 and $A_w=0$, otherwise.
Unfortunately, we do not know how to exclude these possibilities. Some other techniques are needed.

We also found that several quantum codes, such as the parameters $[[14,3,5]]$, $[[17,3,6]]$, $[[17,6,5]]$, $[[19,5,6]]$, $[[21,7,6]]$, $[[23,3,8]]$ and so on,  must have
 $\cS^{\perp}$ with generators of type II, if they exist.
\ee

Next we consider maximal-entanglement EA stabilizer codes, where $\cS_{\text{I}}$ is trivial, in the following  example.
\be \label{ex:linear_programming}
Solving the integer program, we eliminate the existence of
$[[4,2,3;2]]$, $[[5,3,3;2]]$,  $[[5,2,4;3]]$, $[[6,3,4;3]]$,
$[[9,2,7;7]]$, $[[10,2,8;8]]$, $[[10,3,7;7]]$, $[[11,3,8;8]]$,
$[[14,2,11;12]]$, 
$[[15,2,12;13]]$, $[[15,3,11;12]]$,
 $[[15,8,7;7]]$, $[[15,9,6;6]]$,$[[16,4,11;12]]$, $[[17,4,12;13]]$,$[[20,3,15;17]]$ EA stabilizer codes.
The constraints from Theorem~\ref{thm:new constraints} are effective in the case of maximal-entanglement EA stabilizer codes.
\ee
%
%

\subsection{Nonexistence of $[[5,3,3;2]]$, $[[4,2,3;2]]$, $[[5,2,4;3]]$, and $[[6,3,4;3]]$ EA stabilizer Codes}
\label{sec:nonexistence 5332}
A general method to prove or disprove the existence of a code of certain parameters is to use computer search over all possibilities.
It is applicable when the search space is small.
Here we consider a general form of the \emph{check matrix} of quantum codes and use computer search to rule out the existence of some small EA stabilizer codes.

The parity-check matrix  corresponding to a simplified stabilizer group $\cS=\langle g_1,\dots, g_{n-k},h_1,\dots,h_c\rangle $  is defined as
\[
H=\begin{bmatrix}
\tau(g_1)\\
\vdots\\
\tau(g_{n-k})\\
\tau(h_1)\\
\vdots\\
\tau(h_c)
\end{bmatrix}
\]
with  $\text{rank}\left(H\Lambda_{2n} H^T\right)=2c$ \cite{WB08},
where the superscript $T$ means transpose and $\Lambda_{2n}=\begin{bmatrix}0_{n\times n}&I_n\\I_n & 0_{n\times n} \end{bmatrix}$.
Recall that $\tau: {\cG}_{n}\rightarrow \mathbb{Z}_2^{2n}$ is defined in (\ref{eq:tau}).

Theorem 2 in Ref. \cite{LB12} states that a check matrix of a nondegenerate $[[n,k,d;c]]$ code
can be transformed into a standard form $H =  \left[\begin{array}{cc|cc}  I_s &A & D & 0\\ 0&C&I_s&B\\0&E&0&F\end{array}\right]$, where $I_s$ is the $s\times s$ identity matrix and $s\geq d-1$.
Since a maximal-entanglement EA stabilizer code is nondegenerate, we can  construct a check matrix in that form.
The check matrix of an $[[5,3,3;2]]$ code, if it exists, can be written as
\[
H=[H_X|H_Z]=\left[\begin{array}{ccccc|ccccc}
1&0&*&*&*&0&0&*&*&*\\
0&1&*&*&*&0&0&*&*&*\\
0&0&*&*&*&1&0&*&*&*\\
0&0&*&*&*&0&1&*&*&*\\
\end{array}\right],
\]
where a $``*"$ can be $0$ or $1$.
The error syndromes of single-qubit Pauli errors $X_i$ and $Z_j$  are the $i$th column and $j$th column of $H_X$ and $H_Z$, respectively.
If the syndromes of $X_i$ and $Z_i$ are $s^x_i,$ and $s^z_i$, respectively,
then the syndrome of $Y_i$ is $s^x_i+s^z_i$.
Our goal is to fill in the missing columns such that  $\text{rank}(H\Lambda H^{T})=4$ and each single-qubit Pauli error has a unique error syndrome, since the minimum distance is three~\cite{LL11}.

Let the integer number corresponding to a column vector $[a_0\ a_1\ a_2\ a_3]^T$ be $\sum_{i=0}^3 a_0 2^{3-i}$.
The columns 8, 4, 2, and 1 have appeared in the above standard form, and hence so have the columns $10$ and $5$ (the syndromes of $Y_1$ and $Y_2$).
The remaining candidates are $3,$ $6,$ $7,$ $9,$ $11,$ $12,$ $13,$ $14,$ and $15$.
We group these columns as follows:
\begin{align*}
\begin{array}{lll}
G_1: 3, 12, 15; &G_2: 3, 13, 14; &G_3: 6, 9, 15;\\
G_4: 6, 11, 13; &G_5: 7, 9, 14;  &G_6: 7, 11, 12.
\end{array}
\end{align*}
The three columns in any one of  the six groups
are candidates of the syndromes of $X_i,$ $Z_i,$ and $Y_i$ for a fixed $i$.
We further divide these six groups into two non-overlapping sets:
$S_1=\{G_1, G_4, G_5\}$ and $S_2=\{G_2,G_3,G_6\}$.
To fill in the missing columns of $H$, we first choose a set $S_j$,
and then choose two columns from each of the three groups in $S_j$.
Consequently, the total number of candidates  for $H$ is
\[2\times \left({3\choose 2}\times 2!\right)^3=432.\]
We verified that none of them has $\mbox{rank}(H\Lambda H^T)=4$.
Hence there is no $[[5,3,3;2]]$ EA stabilizer code.

The same technique shows that there are no  $[[4,2,3;2]]$,  $[[5,2,4;3]]$, or $[[6,3,4;3]]$ EA stabilizer codes.

\subsection{Nonexistence of other EA stabilizer Codes}

A maximal-entanglement EA code can also be uniquely defined by a logical group $\cL$.
Like the check matrix, we can define a logical matrix $L$ corresponding to  $\cL$ with
 \[\mbox{rank}(L\Lambda_{2n} L^T)=2k.\]
Thus   maximal-entanglement EA codes are a special case of
classical additive quaternary code \cite{CRSS98,LB10}.
\bl \label{lemma:additive codes}
An upper bound on the minimum distance of a classical $(n,2^{2k})$ additive quaternary code
is an  upper bound on the minimum distance of an $[[n,k;n-k]]$ EA stabilizer code.
\el
A table of  upper bounds on the minimum distance of additive quaternary codes for length $n\leq 13$ is given in \cite{BlokhuisBrouwer04,BEFMP09}.
From that table, we learn that there are no
$[[11,4,7;7]]$, $[[12,5,7;7]]$, $[[12,7,5;5]]$, $[[13,6,7;7]]$, $[[13,7,6;6]]$, or $[[13,8,5;5]]$ EA stabilizer codes.
As pointed out in~\cite{LLGF15}, there is no $[[15,5,9;10]]$ EA stabilizer code since there is no $(15,2^{10},9)$ additive quaternary code~\cite{BBFMP13}.

\bl \label{lemma:check}
An upper bound  on the minimum distance of  a classical $(n,2^{2k})$ additive quaternary code
 is an upper bound on the minimum distance of an $(n+1,2^{2(k+1)})$ additive quaternary code,
and hence an upper bound on the minimum distance of an $[[n+1,k+1; n-k]]$ EA stabilizer code.
\el
Lemma \ref{lemma:check} is because that an $(n+1,2^{2(k+1)})$  code has the same number of parity checks as an $(n,2^{2k})$ code.
Hence there are no
$[[14,7,7;7]]$, $[[14,8,6;6]]$, or $[[14,9,5;5]]$ EA stabilizer codes.  

\subsection{Existence of other EA stabilizer Codes}
We say an $[[n,k,d;c]]$ EA stabilizer code is optimal if $d$ achieves any upper bound on the minimum distance for fixed $n,k,c$.
By constructing check matrices as in the proof of nonexistence of $[[5,2,4;3]]$ EA stabilizer codes in Subsec. \ref{sec:nonexistence 5332},
we found optimal EA stabilizer codes with the following parameters:
$[[10,5,5;5]]$, $[[11,7,4;4]]$,  $[[15,12,3;3]]$, and $[[15,11,4;4]]$.

By the EA stabilizer code construction in \cite{BDM06}, the $[9,4,5]$ linear quaternary code can be transformed into a $[[9,4,5;5]]$ EA stabilizer code, which is optimal.
In the same way, we obtain the following EA stabilizer codes: $[[9,5,4;4]]$, $[[10,4,6;6]]$, $[[11,6,5;5]]$,  $[[11,5,6;6]]$,
$[[11,4,6;7]]$, $[[12,8,4;4]]$,  $[[12,5,6;7]]$, $[[12,2,9;10]]$, $[[13,2,10;11]]$, $[[13,6,6;7]]$, $[[14,7,6;7]]$,
and $[[14,8,5;6]]$.
A $[[14,6,6;8]]$ follows naturally.
By deleting the fourth column and the $18$th column of the check matrix of the $[[14,8,5;6]]$ code, we obtain a $[[13,7,5;6]]$ EA stabilizer code.
From the $[12,6,6]$ linear quaternary code, we can obtain a $[[12,5,6;5]]$ EA stabilizer code.
By adding two more ebits, we obtain a $[[12,5,6;7]]$ EA stabilizer code.

Using the circulant construction in \cite{LB10}, we found an optimal $[[11,3,7;8]]$ EA stabilizer code.
We also obtained a $[[15,5,8;10]]$ and a $[[15,6,7;9]]$ EA stabilizer code,
which improves the lower bounds on minimum distance for fixed $n,k,c$.

The $[[14,3,10;11]]$ EA code was mistakenly ruled out in the previous version. An explicit construction of such a code was later given in [GLM23].


Finally, several optimal maximal-entanglement EA stabilizer codes have been constructed from classical quaternary \emph{zero-radical} codes in~\cite{LLGF15}.

The existence of these codes demonstrate that the corresponding linear programming bounds are tight.

\subsection{Table of Lower and Upper Bounds on the Minimum Distance of Maximal-Entanglement EA stabilizer
Codes}
Combining the results in the previous subsections, we improved the table of lower and upper bounds on the
minimum distance of maximal-entanglement EA stabilizer codes for $n\leq 20$ (\cite{LBW13,LLGF15}) in Table \ref{tb:Bounds}.
\begin{table*}
  \centering
  \begin{tabular}{|c|c|c|c|c|c|c|c|c|c|}
\hline
      $n\backslash k$&    $1$& $2$ & $3$&4&5&6&7&8&9\\
      \hline
      3 &  3&   2&   &      &       &&&&\\
      4 &  3&   2&  1&      &       &&&& \\
      5 &  5&   3&  2&      2&      &       &&&\\
      6 &  5&   4&  3&      2&      1&      &&&\\
      7 &  7&   5&  4&      3&      2&      2&&&\\
      8 &  7&   6&  5&      4&      3&      2&      1&      &\\
      9 &  9&   6&  6&      5&      4&      3&      2&      2&\\
      10&  9&   7&  6&      6&      5&      4&      3&      2&1\\
      11& 11&   8&  7&      6&      6&      5&      4&      3&2\\
      12& 11&   9&  8&      7&      6&      5-6&    {4}&      4&3\\
      13& 13&   10& 9&      8&      7&      6&      5&      4&4\\
      14& 13&   10& 10&      8&      7-8&    6-7&    6&      5&4\\
      15& 15&   11& 10&     9-10&   8&      7-8&    6-7&    6&5\\
      16& 15&   12& 11&     10&     9&      8&      7-8&    6-7&6\\
      17& 17&   13& 12&     11&     9-10&   8-9&    7-8&    7-8&6-7\\
      18& 17&   14& 13&     11-12&  10-11&  9-10&   8-9&    8-9&7-8\\
      19& 19&   14& 13-14&  12-13&  11&     10-11&  8-9&    8-9&8\\
      20& 19&   15& 14&     13&     12&     11-12&  9-10&   8-10&8-9\\
\hline
%
%
%
\hline
      $n\backslash k$& 10& $11$& $12$&13&14&15&16&17&18 \\
      \hline
      11 &  2&       &       &      &       &   & &&\\
      12 &  2&      1&      &       &       &   & &&\\
      13 &  3&      2&      2&      &       &   & &&\\
      14 &  4&      3&      2&      1&      &   & &&\\
      15 &  4&      4&      3&      2&      2&  &   &   &   \\
      16&   5&      4&      4&      3&      2&  1&  &   &   \\
      17&   6&      5&      4&      3-4&    3&  2&  2&  &   \\
      18&   6-7&    5-6&    5&      4&      3&  3&  2&  1&  \\
      19&   7&      6-7&    5-6&    5&      4&  3&  3&  2&  2\\
      20&   7-8&    6-7&    6-7&    5-6&    5&  4&  3&  2&  2 \\
\hline
 \end{tabular}
   \caption{Upper and lower bounds on the minimum distance of any $[[n,k,d;n-k]]$
   maximal-entanglement EA stabilizer codes.
   } \label{tb:Bounds}



\end{table*}

\section{Discussion}

We have discussed general EA quantum codes, including nonadditive codes, and
proposed their algebraic linear programming bounds, including Singleton-type, Hamming-type, and the first-linear-programming-type  bounds.
The degenerate and nondegenerate bounds differ for some $\delta$ when  degeneracy exists.
It is known that degenerate EA stabilizer codes can violate the (nondegenerate) Hamming bound. Can we construct such a family of EA stabilizer codes with $R>0$?
Another interesting question is: are there families of degenerate $[[n,k,d;c]]$ EA quantum codes with $R>0$, $\rho>0$ and $1/3< \delta<0.75$  that violate the conventional first linear programming bound?
We can also consider the case of imperfect ebits, which should be similar to the study in~\cite{ALB16}.
Finally, these results could be strengthened in the case of \emph{linear} EA stabilizer codes~\cite{GL13_PhysRevA.87.032309,LLG15}.

We provided a refined Singleton bound for EA quantum codes that works for $d> (n+2)/2$; however, it is not monotonic and may not fully characterize the case of large $c$.
A better Singleton bound for EA stabilizer codes with $d> (n+2)/2$ remains open.

In the setting of usual EA quantum codes, it is assumed that Bob's qubits are error-free. Thus we chose in Theorem \ref{thm:gen_upper_bound} a polynomial of the form $$f(x,y)=F(x)\prod_{j=1}^c (y-j),$$
where $F(x)$ is a polynomial that is used to derive a certain upper bound for general stabilizer codes.
The case that Bob's qubits are imperfect~\cite{LB12} can be developed similarly to the split bounds for data-syndrome codes~\cite{ALB16},
where two types of errors are considered on two disjoint sets of locations.

The linear programming bounds for small quantum codes are improved. The additional constraints in Theorem \ref{thm:new constraints} are
especially effective for maximal-entanglement EA stabilizer codes, where the generators of a symplectic subgroup or a logical group are of type I or type III.
However, they are not that helpful in the case of standard stabilizer codes, where the normalizer groups often have weight distributions of type II. 
It is possible that the linear programming bounds for standard stabilizer codes can be improved for $n\geq 30$.

The EA stabilizer code table in \cite{LBW13} has been significantly improved.
Most of the check matrices of the EA stabilizer codes constructed in this article
are omitted because of limited space.
All of the gaps between the lower bound and upper bound in Table \ref{tb:Bounds} are now closed
for  $d\leq 5$ or $n\leq 8$.
The grouping techniques used in \cite{BEFMP09} may be generalized
to find upper bounds on classical additive quaternary codes for $n=15$ to $20$, which can be used as bounds on maximal-entanglement EA stabilizer codes by Lemma \ref{lemma:additive codes}.

Other types of split weight enumerators can be introduced into the linear programm for EA stabilizer codes. However, they usually induce too many variables so that the integer program is untraceable when $n$ becomes large.
The refined weight enumerator (\ref{eq:refined}) has already introduced too many variables for large $n$.

The method used in Subsec. \ref{sec:nonexistence 5332}
is difficult to apply for larger codes because the computational complexity grows exponentially.
Perhaps we can construct standard stabilizer codes using this grouping method together with Theorem 2 in~\cite{LB12}.

The approach here can be applied to other type of quantum codes, for example, the data-syndrome quantum codes \cite{ALB14,ALB16,Fuji14},
which are codes in the space $\mathbb{F}_4^n\times \mathbb{F}_2^m$.
Split  weight enumerators for these codes can be derived easily.
We can also apply these techniques to other asymmetric quantum codes.


\section*{Acknowledgment}
We thank Todd Brun for helpful discussions. We are grateful to Markus Grassl for his comments and suggestions that help to improve this paper.


\begin{IEEEbiographynophoto}{Ching-Yi Lai}
(M'14) was born in Taipei, Taiwan. He received his MS in 2006 and BS in 2004 from the Department of Electrical Engineering, National Tsing-Hua University, Taiwan. He received his Ph.D. degree in 2013 from the Communication Sciences Institute, Electrical Engineering Department, University of Southern California. He was a postdoctoral research associate at the Centre for Quantum Software and Information, University of Technology Sydney from December 2013 to July 2015. Currently he is a postdoctoral fellow at the Institute of Information Science, Academia Sinica, Taiwan.

His research interests include quantum coding theory, fault-tolerant
quantum computation, quantum information theory, and  quantum cryptography.
\end{IEEEbiographynophoto}

\begin{IEEEbiographynophoto}{Alexei Ashikhmin}
(F'16) is a Distinguished Member of Technical Staff in
  the Communications and Statistical Sciences Research Department of Bell
  Labs, Murray Hill, New Jersey. His research interests include classical and quantum
  information theory, classical and quantum error correcting codes
  communications theory, massive MIMO systems.

  In 2014 Dr. Ashikhmin received Thomas Edison Patent Award for
  Patent on Massive MIMO System with Decentralized Antennas.
  In 2004 he received the Stephen O. Rice Prize for the best
  paper of IEEE Transactions on Communications.
  In 2002, 2010, and 2011 he received Bell Laboratories
  President Awards for breakthrough research in communication projects.

  From 2003 to 2006, and 2011 to 2014
  Dr. Ashikhmin served as an Associate Editor for IEEE Transactions on
  Information Theory.

Dr. Ashikhmin teaches graduate courses as an adjunct professor at Columbia University and New York University.

\end{IEEEbiographynophoto}

\end{document}